\documentclass[10pt,journal,compsoc]{IEEEtran}

\usepackage{graphicx,color,colortbl,psfrag}
\usepackage[usenames,dvipsnames]{xcolor}
\usepackage{bbding}
\usepackage{subfigure}
\usepackage{tabularx}
\usepackage{multirow}
\usepackage{array}
\usepackage{upgreek}
\usepackage{bigints}
\usepackage{epsfig}
\usepackage{mdwlist}
\usepackage{flushend}
\usepackage{array}
\usepackage{bm}

\usepackage{amssymb}
\usepackage{amsmath}
\usepackage{amsfonts}
\usepackage{mathrsfs}
\usepackage{pifont}

\usepackage{eucal}
\usepackage{eurosym}

\usepackage[bookmarks=true,
			colorlinks=true,
			linkcolor=Blue, 
			urlcolor=Blue,  
			citecolor=Blue]{hyperref}
\usepackage{url}

\usepackage[font=normalsize,
			font=bf,
			labelfont=bf]{caption}

\usepackage[thmmarks,thref,hyperref,amsmath]{ntheorem}
\theoremstyle{numberplain}

\newtheorem*{proof}{Proof}

\newtheorem{theorem}{Theorem}

\DeclareMathOperator*{\argmax}{arg\,max}

\usepackage[ruled]{algorithm}
\usepackage{algpseudocode}

\begin{document}
%
\title{\huge Pricing Average Price Advertising Options When Underlying Spot Market Prices Are Discontinuous}

\author{
Bowei~Chen\IEEEcompsocitemizethanks{\IEEEcompsocthanksitem Bowei Chen is with the Adam Smith Business School, University of Glasgow, Glasgow, UK, G12 8QQ. E-mail: \href{mailto:bowei.chen@glasgow.ac.uk}{bowei.chen@glasgow.ac.uk}},~\IEEEmembership{Member,~IEEE,}
and Mohan Kankanhalli\IEEEcompsocitemizethanks{\IEEEcompsocthanksitem Mohan Kankanhalli is with the School of Computing, National University of Singapore, Singapore 117417. \protect E-mail: \href{mailto:mohan@comp.nus.edu.sg}{mohan@comp.nus.edu.sg}}~\IEEEmembership{Fellow,~IEEE.}
}
%

\markboth{IEEE Transactions on Knowledge and Data Engineering}%
{IEEE Transactions on Knowledge and Data Engineering}


\IEEEtitleabstractindextext{%

\begin{minipage}{1\linewidth}
\begin{abstract}
Advertising options have been recently studied as a special type of guaranteed contracts in online advertising, which are an alternative sales mechanism to real-time auctions. An advertising option is a contract which gives its buyer a right but not obligation to enter into transactions to purchase page views or link clicks at one or multiple pre-specified prices in a specific future period. Different from typical guaranteed contracts, the option buyer pays a lower upfront fee but can have greater flexibility and more control of advertising. Many studies on advertising options so far have been restricted to the situations where the option payoff is determined by the underlying spot market price at a specific time point and the price evolution over time is assumed to be continuous. The former leads to a biased calculation of option payoff and the latter is invalid empirically for many online advertising slots. This paper addresses these two limitations by proposing a new advertising option pricing framework. First, the option payoff is calculated based on an average price over a specific future period. Therefore, the option becomes path-dependent. The average price is measured by the power mean, which contains several existing option payoff functions as its special cases. Second, jump-diffusion stochastic models are used to describe the movement of the underlying spot market price, which incorporate several important statistical properties including jumps and spikes, non-normality, and absence of autocorrelations. A general option pricing algorithm is obtained based on Monte Carlo simulation. In addition, an explicit pricing formula is derived for the case when the option payoff is based on the geometric mean. This pricing formula is also a generalized version of several other option pricing models discussed in related studies~\cite{Black_1973, Merton_1973, Merton_1976,Zhang_1998,Kou_2002,Chen_2015_1}. 
\end{abstract}
\end{minipage}

}

\maketitle

\IEEEdisplaynontitleabstractindextext

%
\IEEEpeerreviewmaketitle

\section{Introduction}
\label{sec:introduction}

\IEEEPARstart{O}NLINE advertising refers to advertising using digital technologies through the Internet, where advertisers can quickly promote product information to the targeted customers. Publishers and search engines usually use two ways to sell advertising inventories like page views (also called \emph{impressions}) or link clicks to advertisers~\cite{Evans_2009}. The most popular way is the sealed-bid auction, such as the Generalized Second Price (GSP) auction~\cite{Edelman_2007_2,Varian_2007} and the Vickrey-Clarke-Groves (VCG) auction~\cite{Parkes_2007}. These auction models have been designed with many desirable economic properties. For example, the GSP auction has a locally Envy-free equilibrium, and the VCG auction is efficient and incentive compatible. However, auction models also have limitations. First, it is difficult for advertisers to predict their campaign costs because competition is not visible and occurs in real time. Competiting advertisers and their bidding strategies may change significantly in sequential auctions. Second, the seller's revenue can be volatile due to the uncertainty in auctions. Also, the \lq\lq{}pay-as-you-go\rq\rq{} nature of auctions does not encourage advertisers' engagement because an advertiser can switch from one advertising platform or marketplace to another in the next bidding at near-zero cost. Guaranteed contracts are an alternative way of selling advertising inventories, which can alleviate the limitations of auctions. Usually, an advertiser negotiates a bulk deal with a seller privately. Guaranteed contracts have been recently studied from a variety of different perspectives. Contributors include~\cite{Constantin_2009,Bharadwaj_2010,Salomatin_2012,Bharadwaj_2012,Turner_2012,Chen_2014_2,Hojjat_2014,Chen_2016}. 
However, guaranteed contracts are less flexible. For example, an advertiser needs to make the full non-refundable payment upfront.


Advertising options are a special kind guaranteed contact, allowing its buyer to pay a small upfront fee in exchange for a priority buying right of targeted advertising inventories in the future. The per-inventory payment in the future is pre-specified according to the targeted inventories -- it can be a fixed cost-per-mille (CPM) for impressions in display advertising or cost-per-click (CPC) for clicks in sponsored search. The upfront fee is called the \emph{option price} and the future per-inventory payment is called the \emph{exercise price}. The future payments are not obligatory, which will be based on the number of future deliveries through option exercising by the buyer. Therefore, the advantages of advertising options are obvious. Compared to auctions, the buyer can guarantee the targeted deliveries in the future within a budget constraint. The prepaid option price functions as an \lq\lq{}insurance\rq\rq{} to cap the cost of advertising. Compared to guaranteed contracts, advertising options give the buyer greater flexibility and more control in advertising as he can decide when and whether to exercise the option. Also, advertising options can be seamlessly integrated with the existing auction models because the option buyer's cost is just the pre-paid option price and he can join advertising auctions if he doesn't want to exercise the purchased option in the future. On the sell side, selling advertising options gives publishers and search engines some upfront incomes apart from real-time auctions. More importantly, they are able to establish a contractual relationship with advertisers, which has great potential to increase the long-term revenue. 

Option pricing refers to the calculation of option price for the given specifications. It contains several building blocks: the modeling of underlying price movement; the formulation of option payoff; and the pricing condition or assumption. Previous studies on advertising options have been restricted to two situations. Firstly, the advertising options are path-independent and their payoffs are calculated based on the value of the underlying spot market price at a specific time point. It should be noted that the underlying spot market price is the winning payment price of target inventories from real-time auctions. Since advertising options allow buyers to buy but not sell, the optimal time of option exercising are the option expiration date~\cite{Wang_2012_1,Chen_2015_1}. This leads to the biased calculation of option payoff towards the terminal value. The second limitation is that previous research assumes that the underlying spot market price follows a continuous stochastic process over time. This assumption is not valid for many online advertising slots~\cite{Chen_2015_1,Chen_2015_2}. For example, price discontinuity such as spikes and jumps can be seen in~Figs.~\ref{fig:stylized_fact_ssp}-\ref{fig:stylized_fact_google}.

This paper presents a robust option pricing framework which can be used for general situations. The following contributions are made. Firstly, the option payoff function is designed based on a variable that measures the average underlying spot market prices over a specific time period rather than a time point. It is thus less biased and gives a better overall measurement on price movement, particularly, if there is any price jumps and spikes. Secondly, we use the power mean to calculate the average value, whose special cases and limiting cases offer several different option payoff structures. Therefore, the studied average price advertising option becomes a generalized framework of those relevant advertising options. Thirdly, jump-diffusion stochastic processes are used to describe the underlying spot market price evolution. They can incorporate several important empirical properties including the price discontinuity. We also summarize the empirical properties of prices from advertising auctions. To the best of our knowledge, this is the very first work that provides a such summary. Finally, we discuss how to effectively price the proposed average advertising option via Monte Carlo simulation and also obtain an explicit solution for a special case which generalizes some option pricing models in the related work.

The rest of this paper is organized as follows. Section~\ref{sec:related_work} provides a literature review on options and discusses their recent applications in online advertising. Section~\ref{sec:advertising_options} introduces the basic concepts and transaction procedures of the proposed average price advertising options. Section~\ref{sec:notations_and_setup} sets up the notations and the building blocks of the option pricing model. In Section~\ref{sec:option_pricing_methods}, we discuss the option pricing framework and our solutions. Section~\ref{sec:experiments} presents our experimental results and Section~\ref{sec:conclusion} concludes the paper.

\section{Related Work}
\label{sec:related_work}
 
Options have been used in many fields. Financial options are an important derivative for investors to speculate on profits as well as to hedge risk~\cite{Shreve_2004_2}. Real options deal with choices about real investments as opposed to financial investments, which have become an effective decision-making tool for business projects planning and corporate risk management~\cite{Boer_2002}. Below we first review the important concepts and models in option pricing theory, and then discuss several previous research on advertising options.

Option pricing can be traced back to Bachelier~\cite{Bachelier_1900} who proposed to use a continuous-time random walk as the underlying process to price a call option written on a stock. Call options are a type of option which allows its buyer to buy the underlying assets. Continuous-time random walk is also called \emph{Brownian motion} or \emph{Wiener process}. It is a continuous-path stochastic process $\{W(t), t \geq 0\}$ which satisfies the following conditions: (i) $W(0) = 0$; (ii) the increment $W(t+dt) - W(t)$ is normally distributed $\mathbf{N}(0, dt)$; and (iii) the increment $W(t+dt) - W(t)$ is independent of $\mathcal{F}_t$, the history of what the process did up to time $t$ (also called the \emph{filtration})~\cite{Shreve_2004_2}. Therefore, a Brownian motion is simultaneously a Markov process and a martingale~\cite{Shreve_2004_2}. These two processes are important tools in option pricing. The former describes a random system that changes states according to a transition rule that only depends on the current state, e.g., $W(t+dt) - W(t)$ is independent of $\mathcal{F}_t$. The latter is the mathematical representation of a player\rq{}s fortune in a fair game. Simply, the expected fortune at some later time is equal to the current fortune, e.g., $\mathbb{E}[W(t + dt)|\mathcal{F}_t] = W(t)$, where $\mathbb{E}[\cdot|\mathcal{F}_t]$ represents the conditional expectation given $\mathcal{F}_t$. Since Brownian motion allows negative values, it was then replaced with a geometric form by Samuelson in 1965~\cite{Samuelson_1965}, called \emph{geometric Brownian motion (GBM)}, where the proportional price changes are exponentially generated by a Brownian motion. It satisfies a stochastic differential equation and has an explicit solution by checking It\^{o}'s stochastic calculus (also called \emph{It\^{o} Lemma})~\cite{Shreve_2004_2}. Based on a GBM, Black and Scholes constructed a replicating portfolio for an option and proposed an option pricing method in 1973~\cite{Black_1973}. In the same year, Merton discussed a similar idea to price an option~\cite{Merton_1973}. Their seminal contributions revolutionized the financial industry and spurred the research in this area. Generally, research on options can be classified into four directions~\cite{Sundaresan_2000,Hobson_2004}: (i) complex underlying stochastic models; (ii) valuation of exotic options; (iii) numerical pricing approaches; and (iv) transaction cost models. Our research in this paper is based on the developments of the first three directions, and several related studies are reviewed as follows.

\begin{figure*}[t]
\centering
\includegraphics[width=0.90\textwidth]{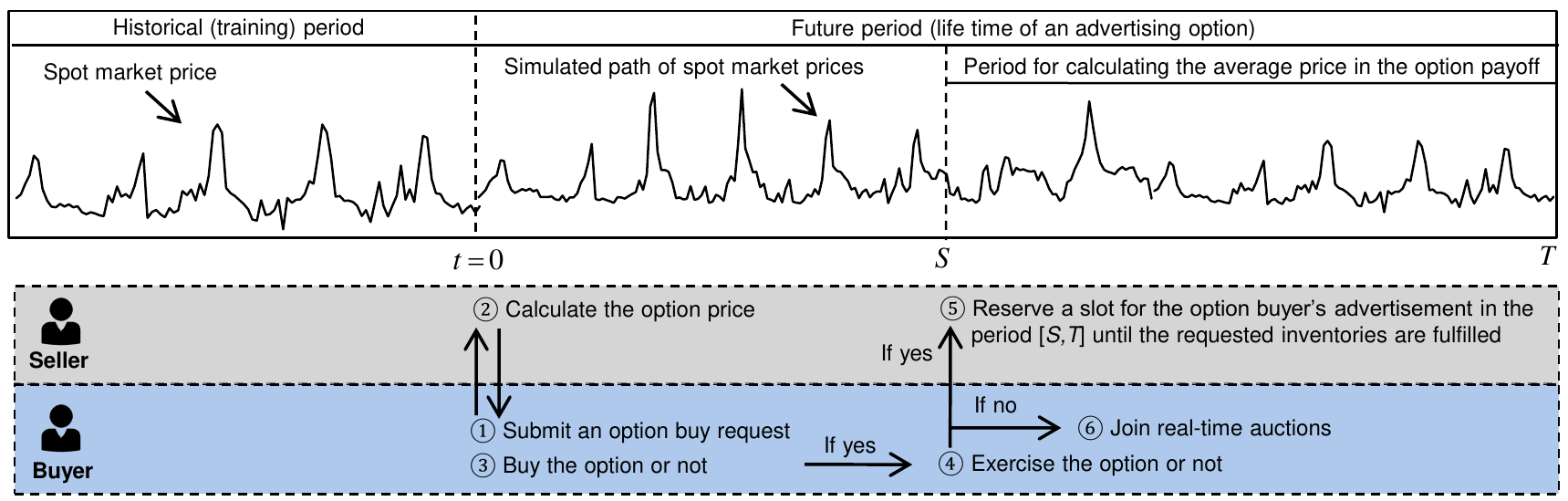}
\caption{Schematic view of buying and exercising an average price advertising option for online advertising.}
\label{fig:advertising_option_system_design}
\vspace*{-12pt}
\end{figure*}

In this paper, we discuss an exotic option tailored to the unique environment of online advertising. Exotic options have been traded for many years in financial markets since the 1980s. In finance, the average price options are one popular type of exotic options, also called \emph{Asian options}~\cite{Zhang_1998}, whose payoff is determined by the average value of prices over a pre-specified period of future time. Therefore, average price options are path-dependent. This is different from the path-independent options such as European options and American options~\cite{Shreve_2004_2}, where the option payoff is calculated for the price at exercise (i.e., on or prior to the option expiration date). The average price options can be divided into two sub-groups: fixed exercise price and floating exercise price. Our proposed average price advertising options are in the former group, where the exercise price is fixed and the random variable is the average underlying price. Several works on of average price options in financial studies are worth mentioning here. A pricing method for financial options whose payoffs are based on a geometric mean was discussed in~\cite{Curran_1994}; an option pricing model for the arithmetic mean case was explored in~\cite{Rogers_1995}; and the power mean option payoff was then discussed in~\cite{Zhang_1998}. These studies offer solid analytical fundamentals for our research. However, they all assume the underlying price movement follows a GBM so that the price needs to be continuous and there are no spikes and jumps. 

%
%
%

Various stochastic processes have been developed over the years for option pricing. Here we focus on jump-diffusion processes which are driven by a Brownian motion and a jump component. Merton~\cite{Merton_1976} proposed a simple stochastic differential equation of jump-diffusion processes, in which jumps are modeled as Poisson events and jump sizes follow a log-normal distribution. Merton's framework was adopted by many other studies and jump sizes can follow different distributions~\cite{Kou_2002,Tsay_2005} which we will discuss in details in Section~\ref{sec:jump_size_dist}. Seasonality has been recently discussed in pricing commodity options like soybeans because price movements in commodity markets often exhibit significant seasonal patterns. They extended the basic GBM structure by incorporating a seasonal behavior parameter which: (i) is determined by a specific function such as sine and cosine
functions~\cite{Barbu_2003}; or (ii) follows another stochastic process driven by an independent randomness~\cite{Richter_2002}; or (iii) or both~\cite{Jin_2010,Backa_2013}. However, these studies didn not incorporate a jump component due to the possible complex model structure. Monte Carlo simulation have been used in option pricing~\cite{Glasserman_2000}. Our this paper develops a method for the exact simulation of continuous-time processes at a discrete set of time steps. The exact simulation means that the joint distribution of the simulated values coincides with the joint distribution of the continuous-time process.

The concept of advertising option was initially proposed in~\cite{Moon_2010}, where a buyer is allowed to make the choice of payment after winning a campaign at either CPM or CPC in the future. This option design is similar to an \emph{option paying the worst and cash}~\cite{Zhang_1998} and the option price is determined by a Nash bargaining game between the buyer and the seller. The first advertising option that allows an advertiser to secure his targeted inventories was discussed in~\cite{Wang_2012_1}. It is a simple European advertising option that considers buying and non-buying the future impressions, and whose price is calculated based on a single-period binomial lattice from a risk-averse publisher's perspective who wants to hedge his expected revenue. This research was then further developed into a multi-period case in~\cite{Chen_2015_2}. Since GBM is not always valid empirically, a stochastic volatility (SV) underlying model was also discussed. In~\cite{Chen_2015_1}, an advertising option with multiple underlying variables following a multivariate GBM was proposed for sponsored search whereby a buyer can target a set of candidate keywords for a certain number of total clicks in the future. Each candidate keyword can also be specified with a unique fixed CPC and the option buyer can exercise the option multiple times at any time prior to or on the its expiration date. This design is a generalisation of the \emph{dual-strike call option}~\cite{Zhang_1998} and the \emph{multi-exercise option}~\cite{Marshall_2012}. Due to the contingent nature of advertising options, they are able to provide greater flexibility to advertisers in their guaranteed deliveries. Our study in this paper is one of the very first studies that discusses contingent payment in online advertising, and the two limitations of the previous studies are addressed by employing a different payoff function (i.e., path-dependent structure) and a different underlying framework (i.e., jump-diffusion stochastic models).

\section{Average Price advertising Options}
\label{sec:advertising_options}

Fig.~\ref{fig:advertising_option_system_design} illustrates the basic concepts, transaction procedures and usage of the proposed average price advertising option in display advertising. A similar scenario can be easily drawn for sponsored search. We assume that a university\rq{}s School of Computer Science creates a new degree program Master of Science in Machine Learning, and is interested in displaying the program\rq{}s banner advertisement online for six months prior to the start of recruitment. As an advertiser, the School will join real-time bidding (RTB) to get the attention of the targeted potential student applicants. However, the banner advertisement can not be guaranteed to be displayed for the needed number of times because the advertising budget is given while the cost of campaigns is uncertain in RTB. To secure the needed future advertising exposure within budget constraint, the School can purchase an advertising option today which gives it a right in the future to obtain the needed impressions at a fixed payment. 

The purchase process of an advertising option has three steps. Step 1 is the School submits a buy request of an advertising option to a publisher at the present time $0$. The request includes: (i) the needed number of impressions from targeted consumers; (ii) the exercise price (i.e., the fixed CPM); and (iii) the future period that the option can be exercised, denoted by $[S,T]$, where $S$ is the time that option can be exercised and $T$ is the option expiration time. Step 2 is the publisher calculates how much to charge the guarantee service upfront for the submitted buy request. This process is called \emph{option pricing} (or \emph{option valuation}). Step 3 is the School decides whether to pay the calculated option price to buy the option. 

We assume the School pays the option price and purchases the option. If the option is exercised by the School at time $S$, the publisher will reserve the specified impressions for the School until the needed number of impressions is fulfilled or the option expires. The School will pay each impression at the pre-specified exercise price. If the option is not exercised, the publisher will not reserve any impressions for the School. Its cost is just the option price so it can still use the remaining budget to join RTB if auctions are preferred in the future. 


\section{Notations and Model Setup}
\label{sec:notations_and_setup}

Both discrete and continuous time notations are used in our model presentations. Discrete time points are denoted by $t_0, \cdots, t_{n}$ where $t_0$ is the present time, $t_{n}$ is the advertising option expiration time, $[t_0, t_n]$ is the life time of the option, $t_{\widetilde{m}}$ is the time that the option buyer can make a decision to exercise the option or not, $[t_{\widetilde{m}}, t_{\widetilde{m} + m}]$ is the period used to calculate the average price in the option payoff and it is also the period that the seller will deliver the requested inventories if the option is exercised at time $t_{\widetilde{m}}$. We denote the continuous time by $t$, $0 \leq t \leq T$. The discrete and continuous time notations have the following relationships: $t_0 = 0$, $t_{\widetilde{m}} = S$, and $t_n = t_{\widetilde{m} + m} = T$. 

\subsection{Jump-Diffusion Stochastic Process}
\label{sec:jump_diffusion_models}

The spot market price of an inventory from a specific advertising slot or from a targeted group of consumers at time $t$ is denoted by $X(t)$. As mentioned earlier, the inventory can be an impression in display advertising or a click in sponsored search. Therefore, $X(t)$ can be expressed as either CPM or CPC. It is the average payment price of the same inventory from the corresponding advertising auctions. Mathematically, the evolution of $X(t)$ can be described by a stochastic process $\{X(t), t \geq 0\}$, which is defined under a filtered probability space $\big(\Omega, \mathcal{F}, \{\mathcal{F}_t\}_{t\geq 0}, \mathbb{P} \big)$, where $\Omega$ is the sample space defining the set of real values that $X(t)$ can take, $\mathcal{F}$ is a collection of subsets of $\Omega$, $\mathbb{P}$ specifies the probability of each event in $\mathcal{F}$, and $\{\mathcal{F}_t\}_{t\geq 0}$ is a filtration satisfying $\mathcal{F}_s \subset \mathcal{F}_t$ for any $0 \leq s < t$. Hence, $X(t)$ is $\mathcal{F}_t$-measurable~\cite{Shreve_2004_2}. 

Given a filtered probability space $\big(\Omega, \mathcal{F}, \{\mathcal{F}_t\}_{t\geq 0}, \mathbb{P} \big)$, $X(t)$ can be modeled by the following stochastic differential equation
\begin{equation}
\label{eq:sde_jump_diffusion}
\frac{d X(t)}{X(t^-)} = 
\underbrace{\mu d t + \sigma d W(t)}_{
\begin{tabular}{c}
\textrm{Continuous}\\ 
\textrm{component}
\end{tabular}} +  \underbrace{d \bigg( \sum_{i=1}^{N(t)} (Y_i - 1) \bigg) }_{
\begin{tabular}{c}
\textrm{Discontinuous}\\
\textrm{component}
\end{tabular}
},
\end{equation}
where $\mu$ is a constant drift, $\sigma$ is a constant volatility, $W(t)$ is a Brownian motion, $X(t^-)$ stands for the value of the spot market price just before a jump at time $t$ if there is one, $N(t)$ represents the arrival of price jumps which is a homogeneous Poisson process with intensity $\lambda$ so that
\begin{align*}
& \mathbb{P}(\textrm{Price jumps once in } dt) = \lambda dt + \mathcal{O}(dt),\\
& \mathbb{P}(\textrm{Price jumps more than once in } dt) = \mathcal{O}(dt),\\
& \mathbb{P}(\textrm{Price does not jump in } dt) = 1 - \lambda dt + \mathcal{O}(dt),
\end{align*}
where $\mathcal{O}(dt)$ is the asymptotic order symbol, and $\{Y_i, i = 1, 2, \cdots\}$ is a sequence of independent and identically distributed (i.i.d.) non-negative variables representing the jump sizes. In the model, all sources of randomness, i.e., $N(t)$, $W(t)$, and $Y_i$, are assumed to be independent. 

There are two major components in Eq.~(\ref{eq:sde_jump_diffusion}). The continuous component is as same as a GBM, in which the drift term represents the expected instantaneous change rate and the volatility term represents the small fluctuation or vibration of price. It has the Markov property. The discontinuous component is driven by a compound Poisson process, which accounts for the unusual or \lq\lq{}abnormal\rq\rq{} extreme changes due to the arrival of some signals in the market. It can also be used to describe the cyclical pattern in price movement. Here we simply explain the term $Y_i - 1$. Let's consider what happens with the spot market price when a jump occurs at time $t_i$, we have
\[
\frac{d X(t_i)}{X(t_i^-)} = \frac{X(t_i) - X(t_i^-)}{X(t_i^-)} = Y_i -1.
\]
Hence, $Y_i =  X(t_i) / X(t_i^-)$, and $Y_i \geq 0$ as $X(t) \geq 0$ for $t \geq 0$. 

Eq.~(\ref{eq:sde_jump_diffusion}) can be solved using It\^{o} stochastic calculus~\cite{Shreve_2004_2}. We now discuss the key steps to the solution. Let $t_1 < t_2 < \ldots $ be the jump times of $N(t)$. For $0 = t_0 \leq t < t_1$, Eq.~(\ref{eq:sde_jump_diffusion}) becomes a GBM because of no jump occurs. By checking It\^{o} Lemma, we obtain
\[
d \ln\{X(t)\} = (\mu - \frac{1}{2} \sigma^2) dt + \sigma d W(t), 
\]
Taking integral of both sides of the equation then gives
\[
X(t) = X(0) \exp \bigg\{(\mu - \frac{1}{2} \sigma^2) t + \sigma W(t) \bigg\}. 
\]
For $t_1 \leq t < t_2$, the solution is similar and we just need to multiply the price with the first jump size $Y_1$. Following the same procedure, the solution to Eq.~(\ref{eq:sde_jump_diffusion}) can be obtained
\begin{align}
X(t) 
= & \ 
X(0) \exp \bigg\{ (\mu - \frac{1}{2} \sigma^2) t + \sigma W(t) \bigg\} \prod_{i=1}^{N(t)} Y_i,
\label{eq:solution_jump_diffusion}
\end{align}
where $\prod_{i=1}^{0} = 1$. It is also an \emph{exponential L\'{e}vy model}~\cite{Cont_2002}. 

\subsection{Jump Size Distributions}
\label{sec:jump_size_dist}

The jump-diffusion stochastic process discussed in Eq.~(\ref{eq:sde_jump_diffusion}) can give several different jump-diffusion stochastic models depending on the jump size distribution. We now discuss three popular choices. 

The first jump size distribution was proposed by Merton~\cite{Merton_1976}, where the logarithm of jump size $V_i = \ln\{Y_i\}$ follows a normal distribution $\mathbf{N}(\alpha, \beta^2)$. This implies 
\begin{align}
\mathbb{E}[e^{V_i}] = & \ e^{\alpha + \frac{1}{2}\beta^2}.
\end{align}

The second jump size distribution was proposed by Kou~\cite{Kou_2002}, where $V_i$ follows an \emph{asymmetric double exponential} distribution, denoted by $\mathbf{ADE}(\eta_1, \eta_2, p_1, p_2)$. Its probability density function is
\[
f_V(v; \eta_1, \eta_2, p_1, p_2) = p_1 \eta_1 e^{-\eta_1 v} \mathbb{I}_{\{v \geq 0\}} + p_2 \eta_2 e^{\eta_2 v} \mathbb{I}_{\{v < 0\}},
\]
where $\mathbb{I}$ is an indicator function, $p_1$ and $p_2$ represent the probabilities of upward and downward jumps (so that $p_1, p_2 \in [0, 1]$ and $p_1 + p_2 = 1$), $\eta_1 > 1$ and $\eta_2 > 0$ are model parameters. The condition $\eta_1 > 1$ is to ensure that $\mathbb{E}(Y_i) < \infty$ and
$\mathbb{E}(X(t)) < \infty$. The model can be rewritten as the combinations of exponentially distributed variables:
\[
V_i =  
\begin{cases}
\varpi_1, & \textrm{with probability } p_1, \\
-\varpi_2, & \textrm{with probability } p_2, 
\end{cases}
\]
where $\varpi_1 \sim \textbf{EXP}(\eta_1)$ and $\varpi_2 \sim \textbf{EXP}(\eta_2)$. Therefore
\begin{align}
\mathbb{E}[e^{V_i}] = & \ p_1 \frac{\eta_1}{\eta_1 - 1} + p_2  \frac{\eta_2}{\eta_2 + 1}.
\end{align}

The third jump size distribution is a special case of Kou's model, where $p_1 = p_2 = \frac{1}{2}$, $\eta_1 = \eta_2 = \frac{1}{\eta}$ and $V$ has the mean $\varrho$. Therefore, $V_i$ follows a \emph{Laplacian} distribution $\mathbf{LAP}(\varrho, \eta)$, and its probability density function is
\begin{align*}
f_V(v; \varrho, \eta) 
= & \ \frac{1}{2 \eta} \exp \bigg\{ - \frac{|v - \varrho|}{\eta} \bigg\}.
\end{align*}
We then have
\begin{align}
\mathbb{E}[e^{V_i}] = & \ \frac{e^{\varrho}}{1-\eta^2}.
\end{align}
For finite samples, the Laplacian distribution is very similar to the Student-t distribution. However, the latter is more tractable analytically and can generate a higher probability concentration such as higher peak around its mean~\cite{Tsay_2005}.

\subsection{Power Mean Option Payoff}
\label{sec:general_mean_option_payoff}

The gain or loss of the buyer of an advertising option in the future is measured by the option payoff. In this paper, the option payoff $\Phi(\mathbf{X})$ is defined as follows

\begin{align}
\Phi(\mathbf{X}) 
= & \
\theta \ \Bigg(
\frac{\widetilde{c}}{c} 
\underbrace{ 
\bigg(\frac{1}{m} \sum_{i = \widetilde{m} + 1}^{\widetilde{m} + m} X_i^\gamma \bigg)^{
\frac{1}{\gamma}}
}_{
\begin{tabular}{l}
$:= \psi(\gamma|\mathbf{X})$ \\
\end{tabular}
}
 - K \Bigg)^+,
\label{eq:option_payoff}
\end{align}
where $(\cdot)^+ := \max\{\cdot, 0\}$, $\theta$ is the requested number of impressions or clicks, $c$ is the CTR of the option buyer's advertisement, $\widetilde{c}$ is the average CTR of relevant or similar advertisements, $K$ is the exercise price which can be a fixed CPM or CPC depending on the advertising type, $\mathbf{X}$ represents a vector of the spot market prices in the future period $[S,T]$, and $\big(\frac{1}{m} \sum_{i = \widetilde{m} + 1}^{\widetilde{m} + m} X_i^\gamma \big)^{1/\gamma}$ is the power mean of these prices.

\begin{table}[t]
\centering
\caption{Special and limiting cases of the power mean.}
\label{tab:general_mean_special_cases}
\begin{tabular}{l|l|l}
\hline
$\gamma$ & $\psi(\gamma|\mathbf{X})$ & Description\\
\hline
$-\infty$ 
& 
\parbox{1.5in}{
$\min\big\{X_{\widetilde{m}+1}, \cdots, X_{\widetilde{m} + m}\big\}$ 
}
& 
Minimum value \\[0.07in]
$-1$ 
&\parbox{1.5in}{
$m/\big(\frac{1}{X_{\widetilde{m}+1}} + \cdots + \frac{1}{X_{\widetilde{m}+m}}\big)$
} 
& 
Harmonic mean \\[0.07in]
$0$ 
&
\parbox{1.5in}{
$\big(\prod_{i=\widetilde{m}+1}^{\widetilde{m} + m} X_i \big)^{\frac{1}{m}}$ 
}
& 
Geometric mean\\[0.07in]
$1$ 
& 
\parbox{1.5in}{
$\frac{1}{m} \sum_{i=\widetilde{m} + 1}^{\widetilde{m} + m} X_i $
}
& 
Arithmetic mean\\[0.07in]
$2$ 
&
\parbox{1.5in}{
$\big(\frac{1}{m} \sum_{i=\widetilde{m} + 1}^{\widetilde{m} + m} X_i^2 \big)^{\frac{1}{2}} $
} 
& 
Quadratic mean\\[0.07in]
$\infty$ 
& 
\parbox{1.5in}{
$\max\big\{X_{\widetilde{m}+1}, \cdots, X_{\widetilde{m}+m}\big\}$
} 
& 
Maximum value\\
\hline
\end{tabular}
\vspace*{-10pt}
\end{table}

The following points are worth noting. First, there are $m$ future spot market prices in the period $[S, T]$, indexed from $\widetilde{m}+1$ to $\widetilde{m} + m$. The index is just to make our notation presentation to be consistent with the previous sections. Second, the term $\widetilde{c}/c$ adds quality effects on the power mean and converts the spot market prices in the period $[S,T]$ into the option buyer's own cost. We can simply consider it as the average payment if he participates in real-time auctions in the same period. Third, if we do not consider CTRs in display advertising, both $c$ and $\widetilde{c}$ can be set to 1. This will not change the option pricing results discussed in the later sections. Fourth, the power mean (also called the \emph{general mean} or the \emph{H\"{o}lder mean}) can be treated as a continuous function of $\gamma$, denoted by $\psi(\gamma|\mathbf{X})$. As shown in Table~\ref{tab:general_mean_special_cases}, it includes the arithmetic mean, the harmonic mean, the quadratic mean and the geometric mean as special cases, and the maximum and minimum observations as limiting cases. Therefore, our option pricing results are a generalization of many previous option pricing models. Details will be discussed in Section~\ref{sec:option_pricing_methods}. In addition, the power mean $\psi(\gamma|\mathbf{X})$ is monotonically increasing. That is, if $\gamma_1 \leq \gamma_2$, then $\psi(\gamma_1|\mathbf{X}) \leq \psi(\gamma_2|\mathbf{X})$. For the detailed proof, see Theorem 6.1 in~\cite{Zhang_1998}. The geometric mean $\psi(0|\mathbf{X})$ is log-normally distributed whereas other means are not. This is an important property so an explicit solution of the option price can be obtained by Theorem~\ref{thm:option_pricing} in Section~\ref{sec:option_pricing_methods}.

\subsection{Arbitrage-Free Condition}
\label{sec:risk_neutral_probability}

The concept of arbitrage is the corner-stone of option pricing theory in finance. Although online advertising and finance have many differences and there is no common marketplace for trading advertising options at the moment, we still believe that ruling out arbitrage opportunities is important and necessary. \emph{Arbitrage} means that an investor can take advantage of a price difference between two or more markets to make profits without taking any risk of loss~\cite{Delbaen_2011}. Simply, his gain happens with probability 1. Arbitrage is the situation that investors can have a \lq\lq{}free-lunch\rq\rq{}. As a consequence, markets will not reach equilibrium. Therefore, it is reasonable to assume there are no arbitrage opportunities among markets so pricing an advertising option needs to be arbitrage-free. 

As previously mentioned in Section~\ref{sec:related_work}, martingale is the mathematical representation of a player\rq{}s fortune in a fair game. To rule out arbitrage opportunities, the discounted spot market price should be a martingale. However, it is not a martingale in the real world. We therefore change the real-world probability measure $\mathbb{P}$ to another \emph{equivalent probability measure} $\mathbb{Q}$ which makes the discounted spot market price a martingale for the purpose of pricing, that is, $\mathbb{E}^{\mathbb{Q}}[e^{-rt}X(t)|\mathcal{F}_0] = X(0)$. We follow the naming convention in finance and call it the \emph{risk-neutral probability}. The idea is best explained by the looking at the following example. We assume there is a demand-side agent who buys impressions for advertisers in display advertising. Here we look at an agent but not an advertiser because we assume that an agent's decision making can be determined by his monetary gains or losses. As guaranteed contracts are usually negotiated privately and there is no disclosed price for a standard guaranteed contract, let's simply consider that a guaranteed contract which specifies a single future impression is sold at its spot market price. In theory, this is achievable in online advertising under the continuous-time setting because the agent can keep buying or selling impressions from real-time auctions over time. At the present time $0$, the agent can have a strategy of borrowing $X(0)$ money from a bank and buying the guaranteed contract. His gain or loss is $\Psi(0) = 0$. If we assume the price space $\Omega = \{\omega_1, \omega_2, \cdots, \}$, then at the future time $t$, the agent can sell the impression to an advertiser at $X(t, \omega)$, and his gain or loss is $\Psi(t, \omega) = X(t, \omega) - X(0)e^{rt}$, where $r$ is a constant continuously compounded interest rate. The term $- X(0)e^{rt}$ represents that the option buyer pays the borrowed money together with the interest back to the bank. Mathematically, arbitrage can be spotted if the following conditions are satisfied: (i) $\Psi(0) = 0$; (ii) $\Psi(t, \omega) \geq 0$ for all $\omega \in \Omega$; (iii) $\Psi(t, \omega) > 0$ for at least one $\omega \in \Omega$. It is not difficult to see that with the risk-neutral probability measure $\mathbb{Q}$, arbitrage opportunities do not exist.

Below we discuss a simple way to find the solution of Eq.~(\ref{eq:sde_jump_diffusion}) under the risk-neutral probability measure $\mathbb{Q}$. 
\begin{align}
 & \ 
\mathbb{E} \bigg[\frac{X(t)}{X(0)}\bigg]
= \mathbb{E} \bigg[\exp\bigg\{ \ln \Big\{\frac{X(t)}{X(0)}\Big\} \bigg\} \bigg] \nonumber\\
= & \ 
\mathbb{E}\bigg[\exp \bigg\{(\mu - \frac{1}{2}\sigma^2)t + \sigma W(t) +\sum_{i}^{N(t)} V_i \bigg\}\bigg] \nonumber\\
= & \ 
\underbrace{
\mathbb{E}\bigg[\exp \bigg\{(\mu - \frac{1}{2}\sigma^2)t + \sigma W(t) \bigg\}\bigg]
}_{
\begin{tabular}{c}
$= e^{\mu t}$
\end{tabular}
} \underbrace{\mathbb{E}\bigg[\exp\bigg\{\sum_{i}^{N(t)} V_i \bigg\}\bigg]}_{
\begin{tabular}{c}
$:= \Lambda$
\end{tabular}
}.
\end{align}

Let $\delta(t)= \sum_{i}^{N(t)} V_i$, then it is a compound Poisson process and $\Lambda$ is the moment generating function (MGF) of $\delta(t)$ at the value $1$. We have
\begin{align}
\Lambda 
= & \ M_{\delta(t)}(1) && \hspace*{-60pt} \triangleright M(\cdot) \textrm{ is a MGF} \nonumber\\
= & \ \sum_{j} e^j \mathbb{P}(\delta(t)= j) \nonumber\\
= & \ \sum_{j} e^j \bigg(\sum_{k} \mathbb{P}\Big(\delta(t)= j | N(t) = k \Big) \mathbb{P}(N(t) = k)\bigg) \nonumber\\
= & \ \sum_{k} \mathbb{P}(N(t) = k) \bigg( \sum_{j} e^j \mathbb{P}\Big(\delta(t)= j | N(t) = k \Big) \bigg) \nonumber\\
= & \ \sum_{k} \mathbb{P}(N(t) = k) \bigg( \sum_{j} e^j \mathbb{P} \bigg( \sum_{i=1}^{k} V_i = j \bigg) \bigg) \nonumber\\
= & \ \sum_{k} \mathbb{P}(N(t) = k) M_{(\sum_{i=1}^{k} V_i)}(1) \nonumber\\
= & \ \sum_{k} \mathbb{P}(N(t) = k) \prod_{i=1}^k \mathbb{E}[e^{V_i}] \nonumber\\
= & \ \sum_{k = 0}^{\infty} \frac{e^{-\lambda t} (\lambda t)^k}{k!} \big( \mathbb{E}[e^{V_i}] \big)^k && \hspace*{-35pt} \triangleright \textrm{ i.i.d. } V_i \nonumber \\
= & \ e^{\lambda t \big(\mathbb{E}[e^{V_i}] - 1 \big)}.
\end{align}

Comparing $\mathbb{E} [\frac{X(t)}{X(0)}]$ and $e^{rt}$ gives $\mu = r - \lambda \big(\mathbb{E}[e^{V_i}] - 1 \big)$. Hence, the solution to Eq.~(\ref{eq:sde_jump_diffusion}) under the risk-neutral probability measure $\mathbb{Q}$ is 
\begin{equation}
\label{eq:solution_risk_neutral}
X(t) = X(0) \exp \bigg\{(r - \lambda \zeta - \frac{1}{2}\sigma^2) t + \sigma W(t) \bigg\}  \prod_{i=1}^{N(t)}Y_i,
\end{equation}
where $\zeta := \mathbb{E}[e^{V_i}] - 1$, and its detailed calculation is given in Table~\ref{tab:measure_change}. 

\begin{table}[t]
\centering
\caption{Calculation of $\zeta$.}
\label{tab:measure_change}
\begin{tabular}{l|l|l}
\hline
Distribution of $Y_i$ & Distribution of $V_i$ & $\zeta$\\
\hline
Log-normal
&
$\mathbf{N}(\alpha, \beta^2)$ 
&
$e^{\alpha + \frac{1}{2}\beta^2}-1$
\\
Log-ADE
&
$\mathbf{ADE}(\eta_1, \eta_2, p_1, p_2)$
&
$p_1 \frac{\eta_1}{\eta_1 - 1} + p_2  \frac{\eta_2}{\eta_2 + 1}-1$
\\
Log-laplacian
&
$\mathbf{LAP}(\varrho, \eta)$ 
&
$\frac{e^{\varrho}}{1-\eta^2}-1$
\\
\hline
\end{tabular}
\vspace*{-5pt}
\end{table}

\section{Option Pricing}
\label{sec:option_pricing_methods}

We now discuss how to price an advertising option. As mentioned in Section~\ref{sec:risk_neutral_probability}, it is possible to construct a replicated strategy for an agent and use it to price an advertising option. Below we discuss a simple method by employing the concept of net present value (NPV), in which incoming and outgoing cash flows can also be described as option benefit and cost, respectively. The benefit of an advertising option is the expected payoff which represents its buyer's relatively cost reduction, and the option cost is the upfront option price. Therefore, we have
\begin{align*}
& \ \ \textrm{NPV (Option)} = \textrm{PV(Option benefit)} - \textrm{PV(Option cost)}.
\end{align*}
We assume that an advertising option adds no monetary value to both buyer and seller so that $\textrm{NPV} = 0$. Then, the option price can be obtained as follows
\begin{equation}
\label{eq:option_pricing_martingale}
\pi_0 = e^{-r T} \mathbb{E}^{\mathbb{Q}} [\Phi (\mathbf{X}) | \mathcal{F}_0],
\end{equation}
where $\mathbb{E}^\mathbb{Q}[\cdot|\mathcal{F}_0]$ represents the expectation conditioned on the information up to time $0$ under the risk-neutral probability measure $\mathbb{Q}$. The option price is the discounted value of the conditional expectation of the option payoff under the risk-neutral probability measure and the option payoff is based on the average mean of the future spot market prices described by the jump diffusion stochastic process.

\subsection{General Solution}

Algorithm~\ref{algo:option_pricing_mc} presents a general solution to Eq.~(\ref{eq:option_pricing_martingale}) using Monte Carlo simulation~\cite{Glasserman_2000}. The time interval $[S,T]$ has been divided into $m$ equal sub-periods $\Delta t$ for a sufficiently large averaging observation. The steps for the period $[0,S]$ are then $\widetilde{m} = \lceil \frac{S}{\Delta t}\rceil$. Hence, the total number of steps in the period $[0, T]$ is $\widetilde{m}+m$, and we denote $t_i = \frac{i}{\widetilde{m} + m} T$, $i =1, \cdots, \widetilde{m}+m$. Let us do Monte Carlo replications for $z$ times. For each $j = 1, \cdots, z$, run the sub-procedures for $\widetilde{m}+m$ steps. For $i =1, \cdots, (\widetilde{m}+m$), the step $\xi_i$ follows a Bernoulli distribution because in a very small time period, no more than one jump can occur almost surely. The confidence bounds are then $\pi_0 \pm e^{-rT} 1.96 \sigma^{\{\Phi\}}/\sqrt{z}$, where $\sigma^{\{\Phi\}} = \mathrm{std}[\{\Phi^{\{j\}}\}_{j=1}^{z}]$. The bounds can be reduced by either increasing the number of replications $z$ or by reducing the variance of option payoffs. For the latter, several variance reduction techniques~\cite{Glasserman_2000} can be used but we do not further discuss them here.  

\begin{algorithm}[t]
\caption{\textbf{Average price advertising option pricing}}
\label{algo:option_pricing_mc}
\noindent\textbf{Input}: $X(0), r, \sigma, S, T, m, K, c, \widetilde{c}, z, \theta, \gamma,\mathbf{\Upsilon}$\\
\noindent (where $\mathbf{\Upsilon} = \{\alpha, \beta\} \textrm{ or } \{\eta_1, \eta_2, p_1, p_2\} \textrm{ or } \{\varrho, \eta\}$)
\begin{algorithmic}[1]
\State $\Delta t \leftarrow \frac{1}{m}(T-S)$; 
\State $\widetilde{m} \leftarrow \lceil \frac{S}{\Delta t}\rceil$; 
\State $\zeta \leftarrow$ Table~\ref{tab:measure_change};
\For{$\; j \leftarrow 1$ to $z$}
	\State $X_0^{\{j\}} \leftarrow X(0)$;
	\For{$\; i \leftarrow 1$ to $\widetilde{m}+m$}
		\State $a_i \leftarrow \mathbf{N} \big( ( r - \lambda \zeta - \frac{1}{2} \sigma^2) \Delta t, \sigma^2 \Delta t \big)$;
		\State $\xi_i \leftarrow \mathbf{BER}(\lambda \Delta t)$;
		\State $v_i \leftarrow \mathbf{N} (\alpha, \beta^2)$;
		\State {\color{white} $v_i \leftarrow$} or $\mathbf{ADE}(\eta_1, \eta_2, p_1, p_2)$ or $\mathbf{LAP}(\varrho, \eta)$;		
		\State $\ln\{X_{i}^{\{j\}}\} \leftarrow \ln\{X_{i-1}^{\{j\}}\} + a_i + \xi_i v_i$;
	\EndFor
\State $\Phi^{\{j\}} \leftarrow$ Eq.~(\ref{eq:option_payoff});
\EndFor
\State $\pi_0 \leftarrow e^{-r T} \big( \frac{1}{z} \sum_{j = 1}^{z} \Phi^{\{j\}} \big)$.
\end{algorithmic}
\noindent\textbf{Output}: $\pi_0$
\end{algorithm}

\subsection{Special Case}

Theorem~\ref{thm:option_pricing} discusses an explicit solution for the case when $\gamma=0$ and $V_i \sim \mathbf{N}(\alpha, \beta^2)$. It has also two special cases. Firstly, if the price averaging period is very short, we can simply use the terminal spot market price as the average price. In this case, the option price can be calculated by using Merton\rq{}s option pricing model~\cite{Merton_1976}. The second special case is that when the jumps sizes are not significant. For example, the estimated value of parameter $\lambda$ or $\alpha$ is very small from training data. Hence, the discontinuous jump component can be removed and the underlying dynamic then becomes a GBM. This the makes our pricing framework similar to an European geometric Asian call option~\cite{Zhang_1998}.

\begin{theorem}
\label{thm:option_pricing}
If $\gamma = 0$ and $V_i \sim \mathbf{N}(\alpha, \beta^2)$, the option price $\pi_0$ can be obtained by the formula
\begin{align}
\pi_0 
= & \ 
\theta e^{- (r+\lambda) T} \sum_{k=0}^{\infty} \frac{(\lambda T)^k}{k!}
\Bigg( \frac{\widetilde{c}}{c} X(0) \Omega \mathscr{N}(\xi_1) - K \mathscr{N}(\xi_2) \Bigg),
\label{eq:option_pricing_special_case}
\end{align}
where $\mathscr{N}(\cdot)$ is the cumulative standard normal distribution function, and
\begin{align*}
A = & \ \frac{1}{2} (r - \lambda \zeta - \frac{1}{2}\sigma^2) (T+S) + k \alpha,\\
B^2 = & \ \frac{1}{3} \sigma^2 T + \frac{2}{3} \sigma^2 S + k \beta^2, \\
\Omega = & \ e^{\frac{1}{2}(B^2+2A)}, \ \ \ \ \ 
\phi =  \ln \{c K\} - \ln \{\widetilde{c} X(0)\},\\
\xi_1 = & \ B - \frac{\phi}{B} + \frac{A}{B}, \ \ \ \ \ 
\xi_2 = \frac{A}{B} - \frac{\phi}{B}.
\end{align*}
\end{theorem}

\begin{proof}
The geometric mean $\psi(\gamma=0|\mathbf{X})$ can be rewritten in a continuous-time form 
\[
\psi(\gamma=0|\mathbf{X}) = \exp \Bigg\{ \frac{1}{T - S} \int_{S}^{T} \ln\{X(t)\} dt \Bigg\}, 
\]
then
\[
Z(T) | {N(T)=k} \sim \mathbf{N} \Big((r - \lambda \zeta -\frac{1}{2}\sigma^2)T + k \alpha, \sigma^2 T + k \beta^2 \Big).
\]
Below we show $\psi(0|\mathbf{X})$ is log-normally distributed.
\begin{align*}
  & \ \psi(0 | \mathbf{X}) 
= X_{0}\bigg(\prod_{i=\widetilde{m}+1}^{\widetilde{m}+m} X_i / X_{0}^m \bigg)^{1/m}\\
= & \ X_{0} 
\exp\bigg\{
\frac{1}{m} \ln\Big\{ 
\Big(\frac{X_{\widetilde{m}}}{X_{0}}\Big)^m
\Big(\frac{X_{\widetilde{m}+1}}{X_{\widetilde{m}}}\Big)^m
\hspace{-10pt}
\cdots
\hspace{5pt}
\Big(\frac{X_{\widetilde{m}+m}}{X_{\widetilde{m}+m-1}}\Big)  
\Big\}
\bigg\}
\end{align*}  
Since $\Delta t = \frac{T-S}{m}$, so $\widetilde{m} = \frac{S}{\Delta t} = \frac{S}{T-S} m$, and then 
\begin{align*}
& \ln\Big\{ \frac{X_{\widetilde{m}}}{X_0} \Big\} \bigg| _{N(T) = k} \hspace{-12pt}  \sim \mathbf{N} \Big((r - \lambda \zeta - \frac{1}{2}\sigma^2) S + k \alpha, \sigma^2 S + k \beta^2 \Big),
\end{align*}
and for $i = 0, \cdots, (m-1)$, 
\begin{align*}
& \ln\Big\{ \frac{X_{\widetilde{m}+i+1}}{X_{\widetilde{m} + i}} \Big\} \bigg| _{N(T) = k} \hspace{-20pt} \sim \mathbf{N} \Big((r-\lambda \zeta - \frac{1}{2}\sigma^2)  \Delta t, \sigma^2 \Delta t \Big). 
\end{align*}
Let $\Theta = \frac{1}{T - S} \int_{S}^{T} Z(t) dt$, then $\Theta \big| _{N(T) = k} \sim \mathbf{N} \big(\widetilde{A}, \widetilde{B}^2 \big)$, where 
\begin{align*}
\widetilde{A} = & \ (r - \lambda \zeta - \frac{1}{2}\sigma^2) ( \frac{(m+1)}{m} \frac{T-S}{2} + S ) + k \alpha, \\
\widetilde{B}^2 = & \frac{(m+1)(2m+1)}{6 m^2}\sigma^2(T-S) + \sigma^2 S + k \beta^2.
\end{align*}
If $m \rightarrow \infty$, $\Theta \big| _{N(T) = k} \sim \mathbf{N}\big(A, B^2\big)$, where 
\begin{align*}
A = & \ \frac{1}{2} (r - \lambda \zeta - \frac{1}{2}\sigma^2) (T+S) + k \alpha, \\
B^2 = & \ \frac{1}{3} \sigma^2 T + \frac{2}{3} \sigma^2 S + k \beta^2.
\end{align*}
Hence, the option price can be obtained as
\begin{align*}
\pi_0 
= & \ \theta e^{-r T} \mathbb{E}^{\mathbb{Q}} 
\Bigg[ 
\mathbb{E}^\mathbb{Q}
\bigg[
\bigg( \frac{\widetilde{c}}{c} X_0 e^{\Theta} - K \bigg)^{+}
\Big| \
N(T) = k
\bigg]
\bigg| \mathcal{F}_0 
\Bigg] \nonumber \\
= & \ \theta e^{-r T} \sum_{k = 0}^{\infty} \frac{(\lambda T)^k}{k!}e^{-\lambda T} 
\mathbb{E}^\mathbb{Q}_0
\bigg[
\bigg( \frac{\widetilde{c}}{c} X_0 e^{\Theta} - K \bigg)^{+}
\bigg] \nonumber \\
= & \ \theta  e^{-r T} \sum_{k = 0}^{\infty} \frac{(\lambda T)^k}{k!}e^{-\lambda T} 
\hspace{-5pt}\int_{\phi}^{\infty} 
\hspace{-6pt}\Big( \frac{\widetilde{c}}{c} X_0 e^{\Theta} - K \Big) f(\Theta) d \Theta,
\end{align*}
solving the integral terms then completes the proof. \hfill $\square$
\end{proof}

\begin{table*}[t]
\centering
\caption{Summary of datasets.}
\label{tab:datasets}
\begin{tabular}{l|l|l|l}
\hline
Dataset   & SSP & Google UK & Google US \\
\hline
Advertising type             & Display      & Search 		& Search\\
Auction model     	& SP      		& GSP  			& GSP\\
Advertising position         & NA     		& 1st position\textsuperscript{$\dagger$}  & 1st position\textsuperscript{$\dagger$} \\
Bid quote         			 & GBP/CPM  	& GBP/CPC 	    & GBP/CPC\\
Market of targeted users\textsuperscript{$\ddagger$} & UK 	        & UK	        & US \\
Time period     & 08/01/2013 - 14/02/2013 & 26/11/2011 - 14/01/2013  & 26/11/2011 - 14/01/2013\\
Number of total advertising slots        & 31          & 106      & 141 \\
Data reported frequency           	  & Auction  & Day   & Day  \\
Number of total auctions      & 6,646,643     & NA       & NA  \\
Number of total bids          & 33,043,127    & NA       & NA \\
\hline
\multicolumn{4}{l}{\textsuperscript{$\dagger$}In the mainline paid listing of the SERP. \textsuperscript{$\ddagger$}Market by geographical areas.}
\end{tabular}
\vspace*{12pt}
\caption{Experimental settings and statistical investigation of stylized facts from the training data.}
\label{tab:summary_stylized_facts}
\begin{tabular}{l|l|l|l|l|l|l|l|l}
\hline
\multicolumn{2}{l|}{ }  & \multicolumn{5}{c|}{SSP} & Google UK & Google US\\
\hline
\multicolumn{2}{l|}{Time scale}   & 1 hour & 4 hours & 6 hours & 12 hours & 1 day & 1 day & 1 day\\
\multicolumn{2}{l|}{$\Delta t$}   & 1.1416e-04 & 4.5662e-04 & 6.8493e-04 & 0.0014 & 0.0027 & 0.0027 & 0.0027\\
\multicolumn{2}{l||}{Data size on each advertising slot:} & 
& & & & & &\\
\multicolumn{1}{l}{} & \multicolumn{1}{l|}{Training set} &
60  &  40 & 30 & 20 & 14 & 60 & 60\\
\multicolumn{1}{l}{} & \multicolumn{1}{l|}{Development and test set} &
60  &  20 & 15 & 5 & 1 & 60 & 60\\
\multicolumn{2}{l|}{Number of selected advertising slots} & 
31  & 23 & 22 & 20 & 12 & 106 & 141 \\
\multicolumn{2}{l|}{$S$} & 
0.0034 & 0.0046 & 0.0034 & 0.0041 & 0  & 0.0822 & 0.0822 \\
\multicolumn{2}{l|}{$T$} & 
0.0068 & 0.0091 & 0.0103  &  0.0068  &  0.0027  & 0.1644 & 0.1644 \\
\hline
\multicolumn{2}{l|}{Jumps and spikes:} &
& & & & & &\\
\multicolumn{1}{l}{} & \multicolumn{1}{l|}{Presence\textsuperscript{$\dagger$}} &
100.00\% & 100.00\% & 100.00\% & 100.00\% & 100.00\% & 100.00\% & 100.00\%
\\
\multicolumn{1}{l}{} & \multicolumn{1}{l|}{Jump size} &
9.43 & 6.46 & 4.84 & 2.76 & 1.48 & 8.46 & 9.15
\\
\hline
\multicolumn{2}{l|}{Normality:} &
& & & & & &\\
\multicolumn{1}{l}{} & \multicolumn{1}{l|}{Kolmogorov-Smirnov test\textsuperscript{$\dagger$}} &
12.90\% & 82.61\% & 86.36\% & 85.00\% & 33.33\% & 0.00\% & 0.00\%
\\
\multicolumn{1}{l}{} & \multicolumn{1}{l|}{Shapiro-Wilk test\textsuperscript{$\dagger$}} &
12.90\% & 78.26\% & 27.27\% & 55.00\% & 91.67\% & 0.00\% & 0.00\%
\\
\hline
\multicolumn{2}{l|}{Heavy tails:} &
& & & & & &\\
\multicolumn{1}{l}{} & \multicolumn{1}{l|}{Presence\textsuperscript{$\dagger$}}         &
100.00\% & 65.22\% & 54.55\% & 5.00\% & 41.67\% & 99.06\% & 100\%
\\
\multicolumn{1}{l}{} & \multicolumn{1}{l|}{Kurtois} &
6.78 & 3.58 & 3.39 & 1.75 & 3.01 & 23.17 & 17.37
\\
\hline
\multicolumn{2}{l|}{Autocorrelations:} &
& & & & & &\\
\multicolumn{1}{l}{} & \multicolumn{1}{l|}{Ljung-Box-Q test at lags 5\textsuperscript{$\dagger$}} &
35.48\% & 82.61\% & 0.00\% & 90.00\% & 16.67\% & 16.98\% & 30.50\%
\\
\multicolumn{1}{l}{} & \multicolumn{1}{l|}{Ljung-Box-Q test at lags 10\textsuperscript{$\dagger$}} &
12.90\% & 91.30\% & 90.91\% & 80.00\% & 0.00\% & 12.39\% & 12.00\%
\\
\multicolumn{1}{l}{} & \multicolumn{1}{l|}{Ljung-Box-Q test at lags 15\textsuperscript{$\dagger$}} &
16.13\% & 91.30\% & 90.91\% & 80.00\% & 0.00\% & 11.50\% & 12.67\%
\\
\hline
\multicolumn{2}{l|}{Volatility clustering:} &
& & & & & &\\
\multicolumn{1}{l}{} & \multicolumn{1}{l|}{Ljung-Box-Q test at lags 5 (abs)\textsuperscript{$\dagger$}} &
22.58\% & 17.39\% & 86.36\% & 20.00\% & 0.00\% & 23.01\% & 34.00\%
\\
\multicolumn{1}{l}{} & \multicolumn{1}{l|}{Ljung-Box-Q test at lags 10 (abs)\textsuperscript{$\dagger$}} &
22.58\% & 43.48\% & 90.91\% & 35.00\% & 0.00\% & 14.16\% & 9.33\%
\\
\multicolumn{1}{l}{} & \multicolumn{1}{l|}{Ljung-Box-Q test at lags 15 (abs)\textsuperscript{$\dagger$}} &
19.35\% & 56.52\% & 90.91\% & 35.00\% & 0.00\% &  12.39\% & 8.00\%
\\
\multicolumn{1}{l}{} & \multicolumn{1}{l|}{Ljung-Box-Q test at lags 5 (square)\textsuperscript{$\dagger$}} &
22.58\% & 13.04\% & 72.73\% & 10.00\% & 0.00\% & 20.35\% & 42.67\%
\\
\multicolumn{1}{l}{} & \multicolumn{1}{l|}{Ljung-Box-Q test at lags 10 (square)\textsuperscript{$\dagger$}} &
12.91\% & 13.04\% & 77.27\% & 25.00\% & 0.00\% & 5.31\% & 4.00\%
\\
\multicolumn{1}{l}{} & \multicolumn{1}{l|}{Ljung-Box-Q test at lags 15 (square)\textsuperscript{$\dagger$}} &
6.45\% & 13.04\% & 72.73\% & 25.00\% & 0.00\% & 4.42\% & 5.33\%
\\
\hline
\multicolumn{2}{l|}{Selected advertising slots\textsuperscript{$\ddagger$}}      &
58.06\% & 8.70\% & 0.00\%& 10.00\%  & 100.00\% & 72.64\%  & 75.18\% 
\\
\multicolumn{2}{l|}{}  &
(18)  & (2)  &  (0)   & (2) & (12)  & (77) & (106)\\
\hline
\multicolumn{9}{l}{
\textsuperscript{$\dagger$}The number represent the presence (or acceptance) percentage of advertising slots of the corresponding group.
}\\
\multicolumn{9}{p{6.7in}}{
\textsuperscript{$\ddagger$}Advertising slots with absence of both autocorrelations and volatility clustering are selected for pricing advertising options. The number outside the round brackets represents the percentage of advertising slots which have revenue increase and the number in the round brackets represents the average revenue change of slots under that group.
}\\
\end{tabular}
\vspace*{-10pt}
\end{table*}

\subsection{Properties of the Option Price}

We now discuss what happens to the option price when there is a change to one of the model parameters or factors, with all the other factors remaining fixed. It is obvious that an advertising option becomes more valuable if the present underlying spot market price increases. Since the option payoff measures how does the average spot market price in the future period exceeds the exercise price, the option price will decrease if the exercise price increases. In option pricing, an opportunity cost is involved. This cost depends upon the risk-less bank interest rate, the length of the averaging period and the time to expiration. Increases in any of these three factors will increase the option price. The volatility of underlying spot market prices is a measure of uncertainty in online auctions, which also represents the risk that a publisher or search engine will take if he sells an advertising option. The higher the volatility, the greater will the option price be. The jump related parameters affect the option price in a less clear-cut way. The option price will also increase if more advertising inventories are requested. The market CTR and the option buyer's CTR will consistent with their effects in online auctions. The higher the market CTR, the higher the option price is; the higher the buyer's CTR, the less the option price is. As also noted previously, both CTRs can be set to 1 in display advertising.

\section{Experiments}
\label{sec:experiments}

This section describes our datasets and experimental settings, investigates the statistical properties of spot market prices in real-time advertising auctions, discusses the parameters' estimation for the underlying jump-diffusion stochastic process, and presents the option pricing results and revenue analysis. 

\begin{figure*}[t]
\begin{minipage}{0.49\linewidth}
\centering
\includegraphics[width=1\linewidth]{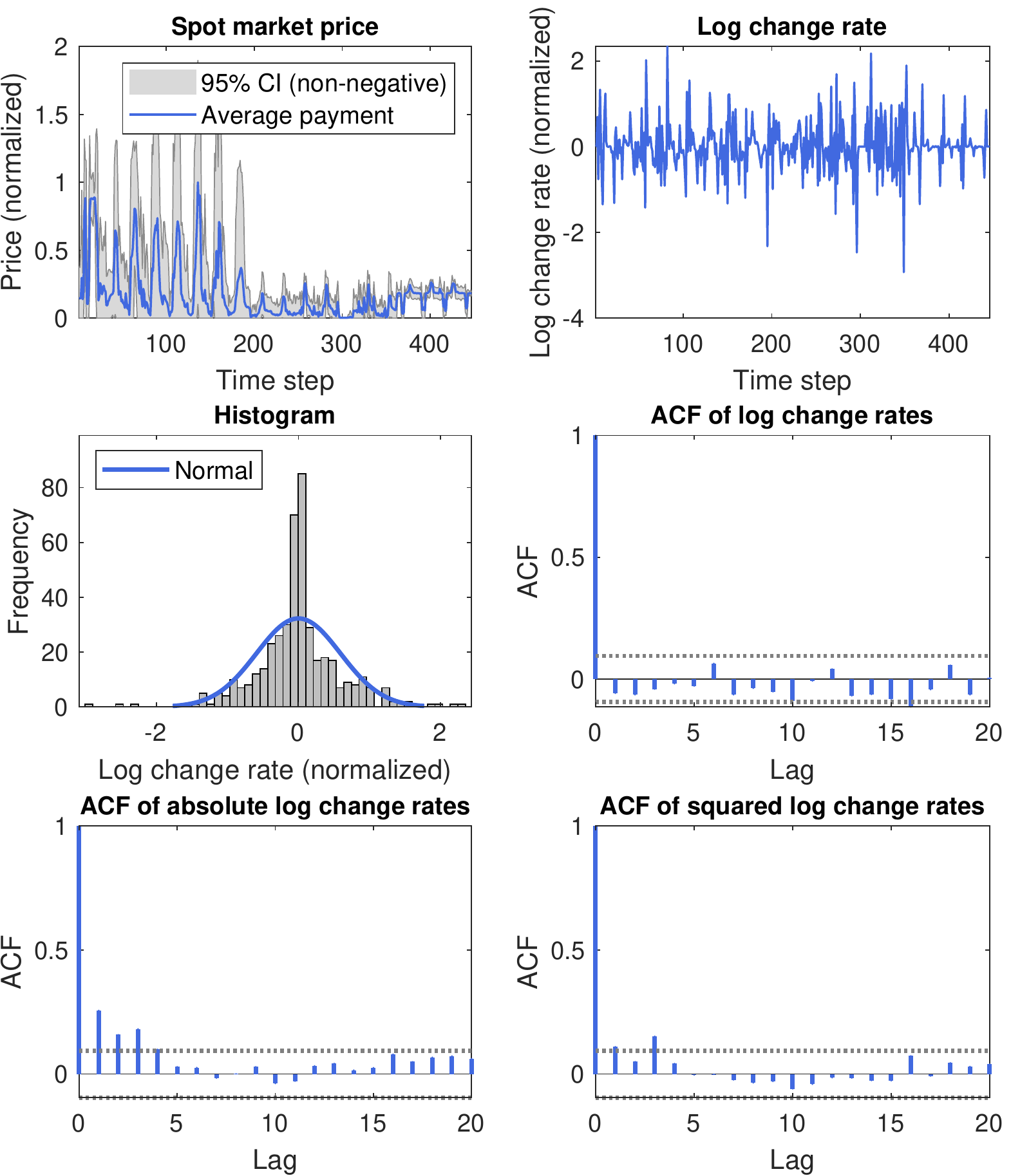}\vspace*{3pt}
\caption{Time series plots and statistical tests of the hourly spot market prices from an advertising slot in the SSP dataset over the period between 08/01/2013 14:00 and 27/01/2013 4:00.}
\label{fig:stylized_fact_ssp}
\vspace*{-8pt}
\end{minipage}
\hspace{8pt}
\begin{minipage}{0.49\linewidth}
\centering
\includegraphics[width=1\linewidth]{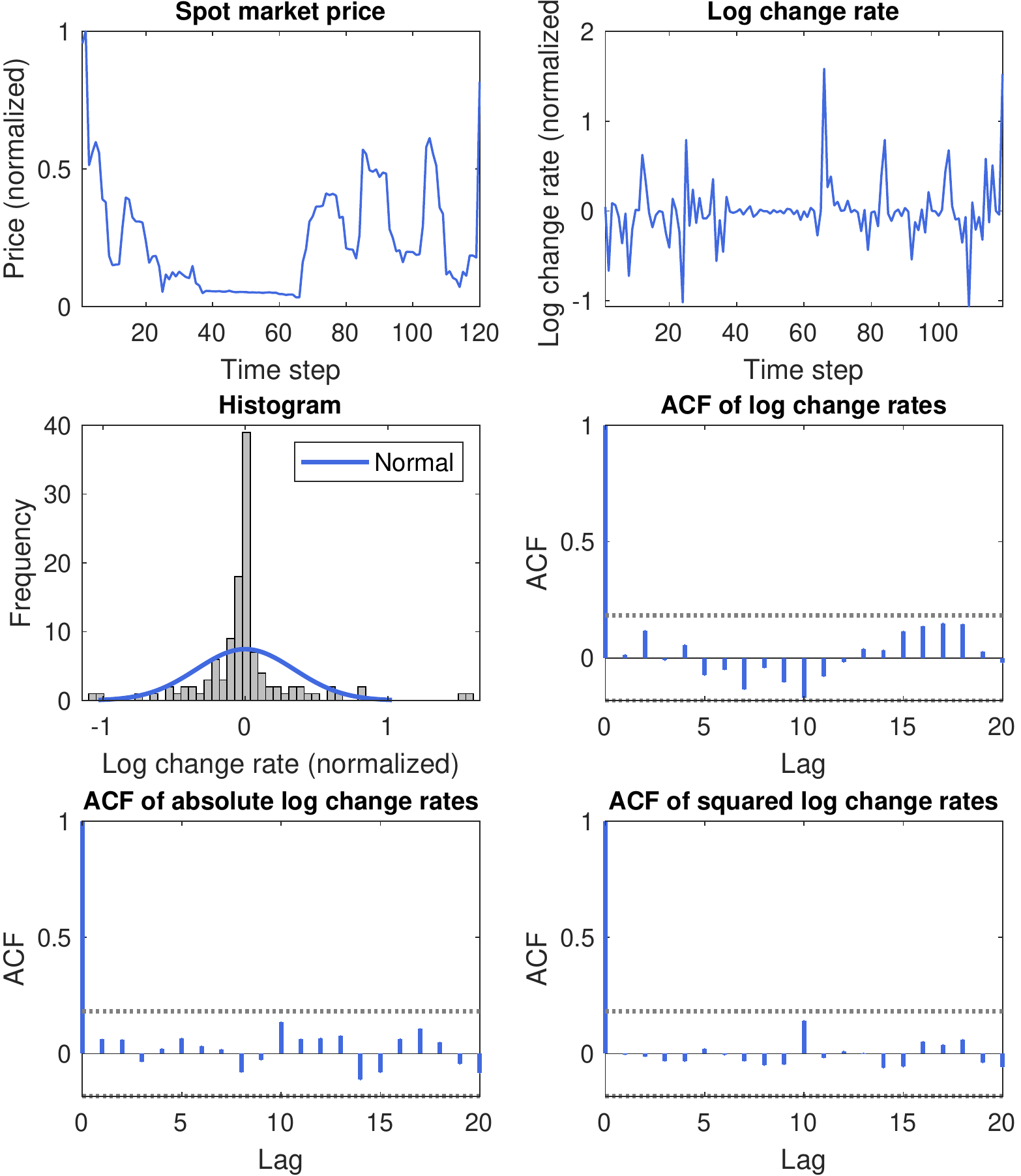}\vspace*{3pt}
\caption{Time series plots and statistical tests of the daily spot market prices of the first position of the mainline paid listing for the keyword \lq\lq{}notebook laptops\rq\rq{} in Google US dataset over the period between 17/03/2012 and 21/09/2012.}
\label{fig:stylized_fact_google}
\vspace*{-8pt}
\end{minipage}
\end{figure*}

\subsection{Data and Experimental Settings}
\label{sec:datasets}

As shown in Table~\ref{tab:datasets}, we use three datasets in experiments: an RTB dataset from a medium-sized supply-side platform (SSP) in the UK and two sponsored search datasets from Google AdWords. The SSP dataset contains 31 advertising slots though we do not know the detailed positions of those slots. Multiple advertising slots (even on a same webpage) are sold separately in RTB through the Second Price (SP) auction model~\cite{Muthukrishnan_2009}. In sponsored search, Google uses keywords to target online users' search queries and use the GSP auction model to sell a list of advertising slots on its search engine result page (SERP) for a query relevant to a specific keyword. Google UK and US datasets report the keyword auctions which target online users from different geographical locations. Google UK dataset contains 106 keywords and Google US contains 141 keywords. As we only look at the first poisiton in the mainline paid listing of SERPs, we could consider the keywords are the unique advertising slots. Our datasets have also been used in several other recent online advertising research. The SSP dataset has been used in~\cite{Chen_2014_2,Chen_2015_2,Yuan_2014,Yuan_2013_2}, and Google\rq{}s datasets have been used in~\cite{Chen_2015_1,Chen_2015_2,Yuan_2012}. The previous studies examined bids in advertising auctions. Spot market prices were extracted from advertising auctions and were used to check the assumptions of the GBM model in~\cite{Chen_2015_1,Chen_2015_2}. In this section, we provide in-depth statistical insights and give a comprehensive investigation of common features of the spot market price. 

Our experimental settings are described in Table~\ref{tab:summary_stylized_facts}. It is worth further explaining two settings here. First, different time scales are used to extract the time series of spot market prices from sequential advertising auctions. As we use the risk-less bank interest rate in the option pricing model, we need to follow the convention of computing time scale in finance. We consider a one-year time period is 1 because the compound risk-less bank interest rate is usually expressed as an annual rate~\cite{Shreve_2004_2}. Therefore, if the time scale $\Delta t$ is chosen as a day, $\Delta t = 1/365 \approx 0.0027$. If the time scale $\Delta t$ is chosen as an hour, $\Delta t = (1/365) \times (1/24) \approx \textrm{1.1416e-04}$. Google's datasets only contain the daily payment prices but we can further extract prices on a smaller time scale (e.g., 12 hours, 6 hours, 4 hours and 1 hour) from the SPP dataset. Second, we randomly select the time period which has consecutively reported data so that the evolution of spot market prices can be analyzed. Also, as described in Section~\ref{sec:notations_and_setup}, the period $[S,T]$ is used to calculate the power mean in the advertising option payoff. If $S=T$, there is only one time point and the option becomes an European call option~\cite{Shreve_2004_2}.

%

\begin{figure}[t]
\centering
\includegraphics[width=1\linewidth]{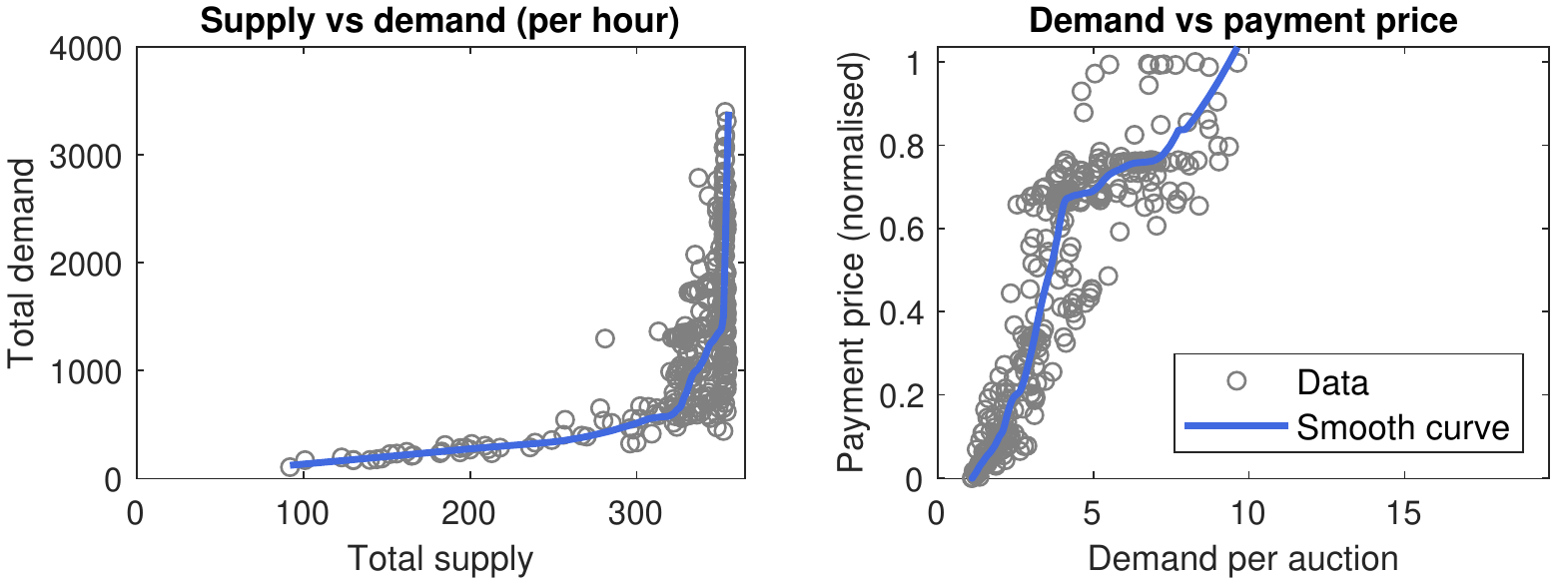}
\vspace*{-7pt}
\caption{Relationships among supply, demand and payment price for an advertising slot in the SSP dataset. The smooth curves are obtained using the 1st degree polynomial locally weighted scatterplot smoothing (LOWESS) method~\cite{Cleveland_1979}.}
\label{fig:supply_demand_ssp}
\vspace*{15pt}
\includegraphics[width=0.75\linewidth]{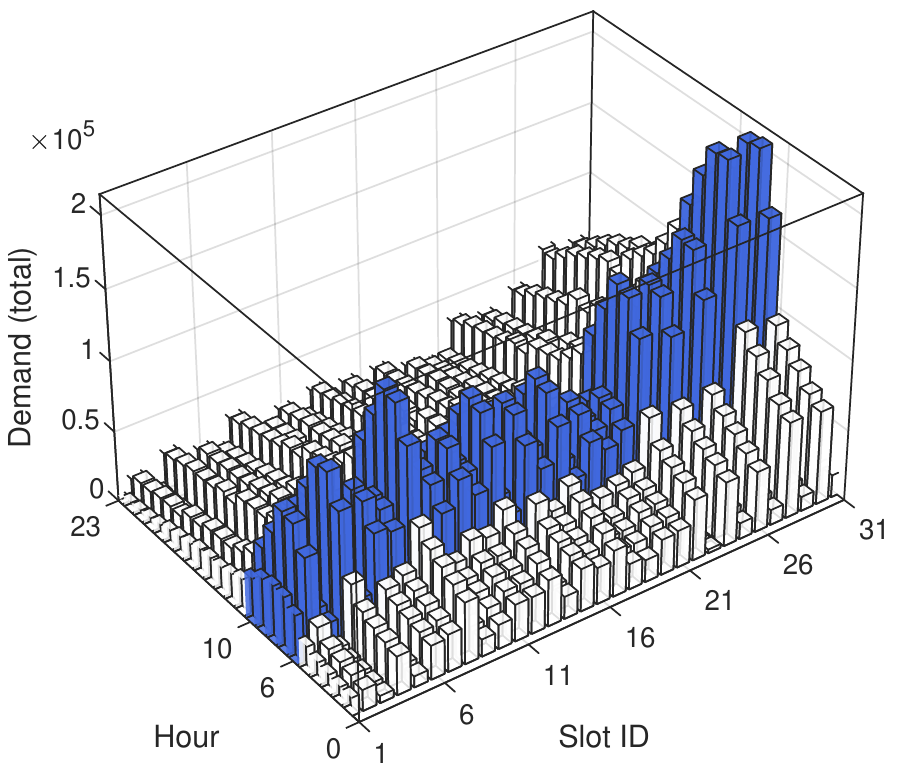}
\caption{Distribution of total advertising demand in the SSP dataset.}
\label{fig:demand_peak_ssp}
\vspace*{-10pt}
\end{figure}

\begin{figure}[t]
\centering
\includegraphics[width=1\linewidth]{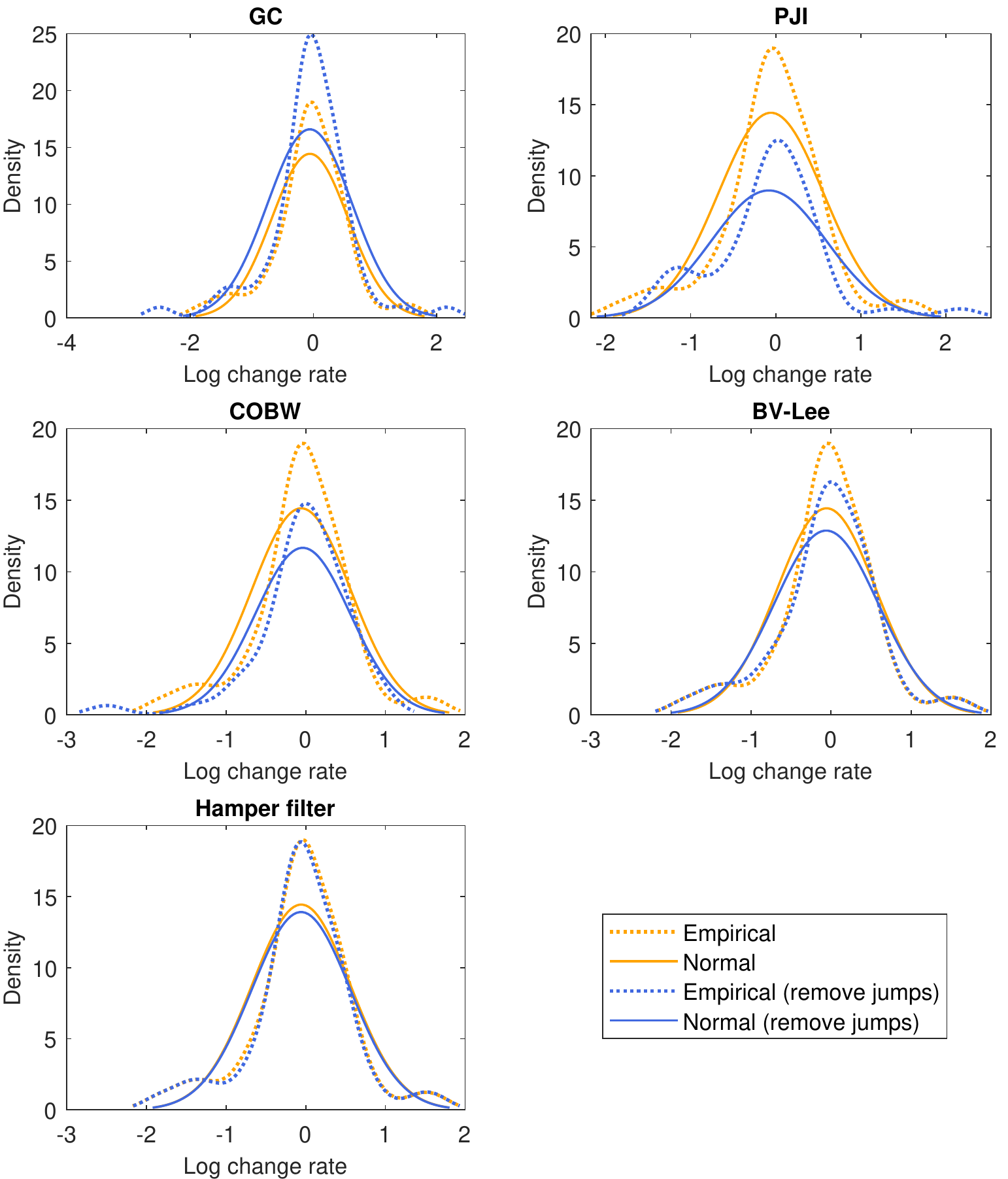}
\caption{Detecting price jumps for an advertising slot in the SSP dataset.}
\label{fig:jumps_detection_example}
\vspace*{-10pt}
\end{figure}

\subsection{Stylized Facts}
\label{sec:stylized_fact}


A set of common features or statistical properties of the spot market price has been identified in our experiments which are known as \emph{stylized facts}~\cite{Sewell_2011}. As $X(t) \geq 0 $, for mathematical convenience, its logarithm is usually analyzed. Given a time scale $\Delta t$, which can range from a few seconds to a day, the log change rate of $X(t)$ at time scale $\Delta t$ is defined as 
\[
R(t, \Delta t) = \ln \{X(t + \Delta t)\} - \ln \{X(t)\}. 
\]
It is also called the \emph{log return} or \emph{continuously compounded return} in time series analysis~\cite{Tsay_2005}. One important reason that the log change rate is used is because it has more tractable statistical properties than the simple rate. For example, a multiple period log change rate can be expressed as a sum of one-period log change rates. 

\textbf{\textsl{Jumps and Spikes}} As shown in Figs.~\ref{fig:stylized_fact_ssp}-\ref{fig:stylized_fact_google}, the spot market prices exhibit sudden jumps and spikes in both display advertising and sponsored search. From a modeling point of view, the price process exhibits a non-Markovian behavior in short time intervals and prices increase or decrease significantly in a continuous way. In our datasets, all advertising slots show price jumps and spikes while they have different jump frequencies. Several jump detection techniques will be discussed in~Section~\ref{sec:para_estimation}. This stylized fact has also been discussed in electricity prices and the typical explanation is a non-linear supply-demand curve in combination with the electricity's non-storability~\cite{Bierbrauer_2007}. This explanation can possibly be applied to online advertising. Fig.~\ref{fig:supply_demand_ssp} gives an example from an advertising slot showing the relationships among supply, demand and payment price. The advertising supply is triggered by ad-hoc web surf or search made by online users, and the demand is the number of media buyers who join advertising auctions. For a specific advertising slot, the competition level in advertising auctions affects the payment price nonlinearly. It is worth noting that the hourly spot market prices from RTB in Fig.~\ref{fig:stylized_fact_ssp} exhibit some cyclical patterns. Fig.~\ref{fig:demand_peak_ssp} explains this by showing the peak hours of the total advertising demand in the SSP dataset. In fact, cyclical bid adjustments were discussed for real-time sponsored search auctions~\cite{Zhang_2011}. Several time series and signal processing methods can be used to decompose or extract the cyclical patterns~\cite{Tsay_2005}. However, we do not further discuss this in this study due to the following reasons. First, such cyclical patterns mainly exist in time series data with updates that occur in less than a day. As shown in Fig.~\ref{fig:stylized_fact_google}, the daily spot market prices do not exhibit cyclical patterns. Second, the proposed jump-diffusion stochastic process is capable of reproducing some cyclical patterns based on the homogeneous Poisson process because jumps and spikes look arrive at a constant rate. Proposing a new model which can accurately incorporate the cyclical patterns as well as being used for pricing advertising options can be an interesting direction of future research.

\textbf{\textsl{Non-Normality and Heavy Tails}} The non-normal character of the unconditional distribution of log change rates has been observed in many advertising slots. Normality can be graphically checked by a histogram or Q-Q plot, and can be statistically verified by hypothesis testing such as the one-sample Kolmogorov-Smirnov (KS) test~\cite{Massey_1951} and the Shapiro-Wilk (SW) test~\cite{Shapiro_1965}. Figs.~\ref{fig:stylized_fact_ssp}-\ref{fig:stylized_fact_google} exhibit that the distributions of log change rates have tails heavier than those of the normal distribution. This is also called \emph{leptokurtic}. One way to quantify the deviation from the normal distribution is checking the kurtosis statistic~\cite{Tsay_2005}. Since the kurtosis of a standard normal distribution is $3$, then the empirical distribution will have a higher peak and two heavy tails if the kurtosis is larger than 3. 

\textbf{\textsl{Absence of Autocorrelations}} Consider whether the future log change rates can be predicted from the current values, we can formulate this question by asking whether they remain stable and whether they are correlated over time~\cite{Cont_2001}. The process is assumed to be weakly stationary so that the first moment and autocovariance do not vary with respect to time. The (linear) autocorrelation functions (ACFs) of log change rates are insignificant in Figs.~\ref{fig:stylized_fact_ssp}-\ref{fig:stylized_fact_google}. Therefore, both processes can be constructed by using the Markov property. In fact, as discussed in Section~\ref{sec:related_work}, most of the classical models in economics and finance assume that asset prices follow a GBM, which is based on independent asset returns~\cite{Samuelson_1965_2}. In experiments, we also use the Ljung-Box Q-test to check the autocorrelation for a fixed number of lags. Table~\ref{tab:summary_stylized_facts} shows that most of hourly and daily log change rates (overall more than 80\%) do not have autocorrelations. However, most of the 4-hour, 6-hour, and 12-hour rates in the SSP dataset exhibit autocorrelations. 


\textbf{\textsl{Volatility Clustering}} A stochastic process can have uncorrelated but not independent increments. The magnitude of price fluctuations is measured by volatility. \emph{Volatility clustering} is referred to the property that large price variations are more likely to be followed by large price variations~\cite{Cont_2001}. To detect volatility clustering, two commonly used methods are: (i) the ACF of absolute log change rates; and (ii) the ACF of squared log change rates. Volatility clustering has been observed in Fig.~\ref{fig:stylized_fact_ssp} but not Fig.\ref{fig:stylized_fact_google}. Table~\ref{tab:summary_stylized_facts} shows that it is not the property for the majority of online advertisements, particularly, for hourly and daily rates.

In this paper, the discussed jump-diffusion stochastic process can incorporate the first three properties but not volatility clustering. This property can be incorporated by adding another dynamic for volatility such as the SV model discussed in~\cite{Chen_2015_2}. In the following experiments, we use hourly and daily data to develop jump-diffusion stochastic models and then use them for pricing advertising options.

\begin{figure}[t]
\centering
\includegraphics[width=1\linewidth]{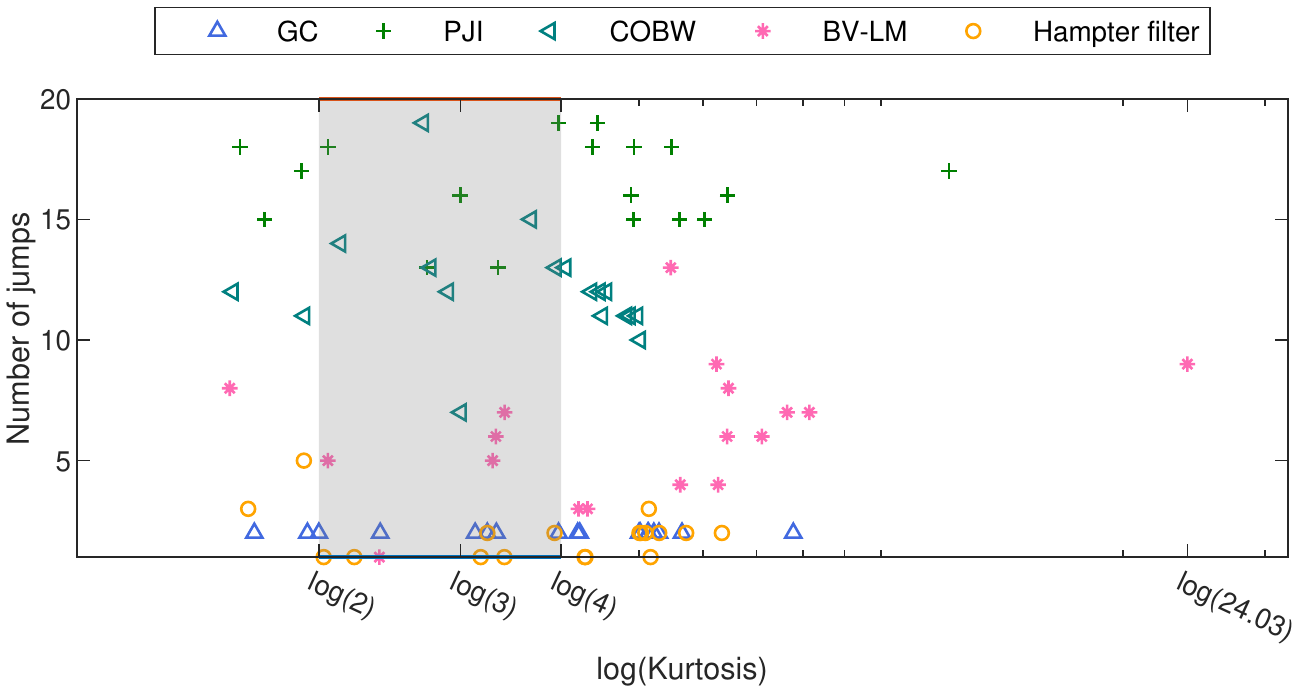}
\caption{Illustration of selecting price jump detection methods on the SSP dataset (hourly time scale).}
\label{fig:jumps_detection_performance_example}
\vspace*{18pt}
\includegraphics[width=1\linewidth]{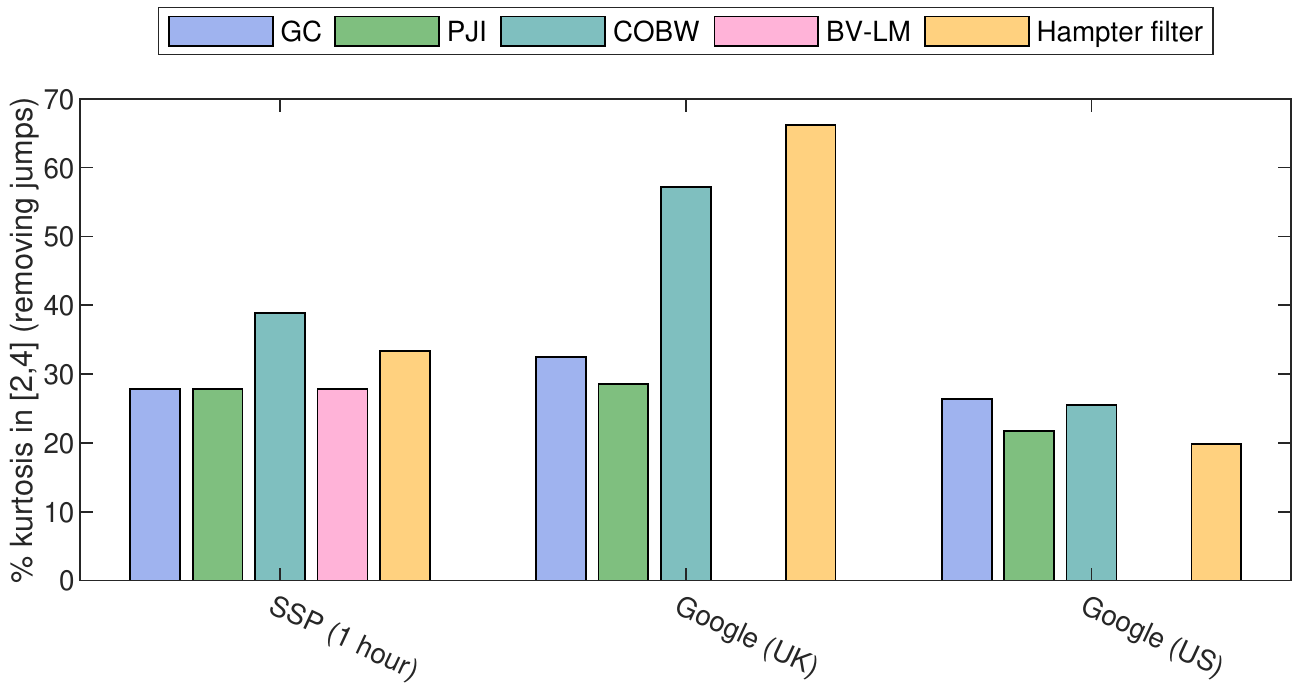}
\caption{Comparison of price jump detection methods on the SSP and Google datasets.}
\label{fig:jumps_detection_performance_summary}
\vspace*{-15pt}
\end{figure}

\begin{figure}[t]
\centering
\includegraphics[width=1\linewidth]{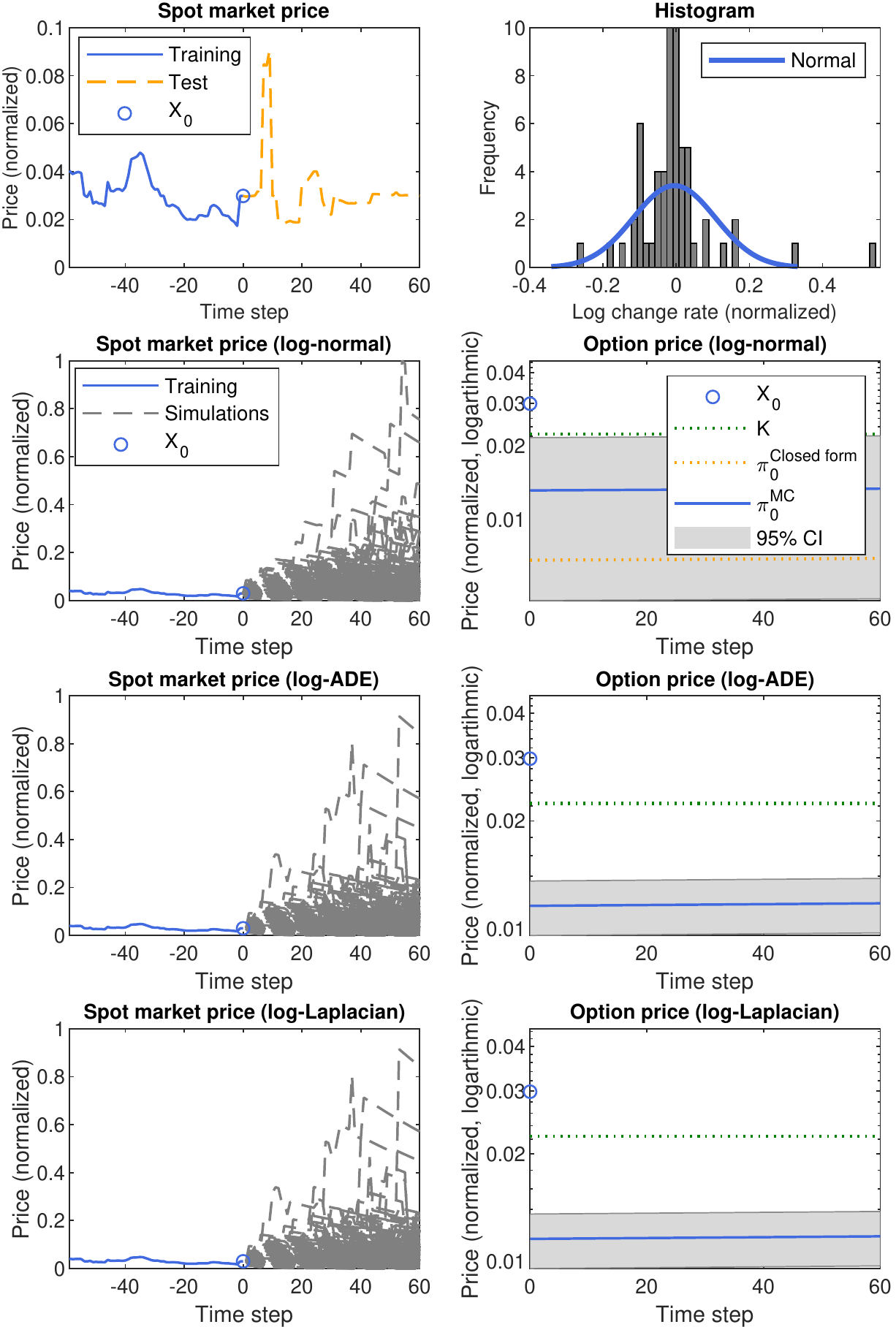}
\caption{Pricing an average price advertising option for the keyword \lq\lq{}panasonic dmc\rq\rq{} in Google UK dataset.}
\label{fig:jump_price_example}
\vspace*{-15pt}
\end{figure}

\subsection{Estimation of Model Parameters}
\label{sec:para_estimation}

One of the widely accepted interpretations of price jumps considers them as time-dependent outliers. Simply, a price jump is an observation that lies in an abnormal distance from other values. In~\cite{Hanousek_2012}, an extensive simulation study was conducted to compare the relative performance of several detection methods for price jumps, including the global centiles (GC), the price-jump index (PJI), the centiles over block-windows (COBW), and various bipower variation methods. The comparison results showed that: (i) the GC and the COBW outperformed others in the case of false positive probability; (ii) the bipower variation method proposed by Lee and Mykland~\cite{Lee_2008} (abbreviated as BV-LM) performed best in the case of false negative probability. In our experiments, we implement all these methods, and the hamper filter~\cite{Pearson_2005} because it has been widely used for outlier detection in signal processing. The identified jumps for the training sets are presented in Table~\ref{tab:summary_stylized_facts}. 

Fig.~\ref{fig:jumps_detection_example} gives an illustration of detecting price jumps on an advertising slot from the SSP dataset. There are large differences in terms of performance among methods. If there is no jump, our discussed stochastic underlying framework in Eq.(\ref{eq:sde_jump_diffusion}) then becomes a GBM. Therefore, after removing the detected jumps, the kurtosis of the training log change rates should approach 3 and we use this creatiera to select the best jump dection method. An illustration of our model selection is described in Fig.~\ref{fig:jumps_detection_performance_example}, where the COBW performances best for the SSP dataset based on hourly time scale as it has the highest number of slots which lie in the range $[2,4]$. Fig.~\ref{fig:jumps_detection_performance_summary} further summarizes the overall results of model selection in our training sets. It should be noted that the BV-LM is not used for Google datasets because there are no intraday campaign records. In summary, after removing the identified jumps, the COBW performs best as it has 40.50\% of slots that the kurtosis of log change rates lies in the range $[2,4]$, followed by the hamper filter with 39.80\%. In the following experiments, jumps are identified by the COBW.

We follow~\cite{Brigo_2007} and estimate other model parameters using the maximum likelihood method. The discretization of Eq.~(\ref{eq:solution_jump_diffusion}) is
\begin{align}
\frac{X(t)}{X(t - \Delta t)} = & \ 
\exp\bigg\{
(\mu - \frac{1}{2} \sigma^2) \Delta t + \sigma \sqrt{\Delta t} \varepsilon_t
\bigg\} \prod_{i=1}^{n_t} Y_i,
\end{align}
where $\varepsilon_t \sim \mathbf{N}(0, 1)$, and $n_t = N(t) - N(t-\Delta t)$ representing the number of price jumps between time $t-\Delta t$ and time $t$. 


Let $\widetilde{Z}(t)= \ln\{X(t)/X(t-\Delta t)\}$, $\mu_V = \mathbb{E}[V_i]$, $\sigma^2_V = \mathrm{Var}[V_i]$, $\mu^* = \mu - \frac{1}{2} \sigma^2 + \lambda \mu_V$, and $\Delta J_t^* = \sum_{i=1}^{n_t} V_i - \lambda \mu_{V} \Delta t$, then we have $\mathbb{E}[\Delta J_t^*] = \mathbb{E}[n_t] \mu_V - \lambda \Delta t \mu_V = 0$, $\mathbb{E}[\Delta J_t^* | n_t] = n_t \mu_V - \lambda \Delta t \mu_V$, $\mathrm{Var}[\Delta J_t^* | n_t] = n_t^2 \sigma_V^2$, $\mathbb{E}[\widetilde{Z}(t) | n_t ] = (\mu - \frac{1}{2} \sigma^2 ) \Delta t + n_t \mu_V$, $\mathrm{Var}[\widetilde{Z}(t) | n_t ] = \sigma^2 \Delta t + n_t^2 \sigma^2_V$. For simplicity, $\widetilde{Z}(t) | n_t$ is considered to be normally distributed, then we can maximize the logarithmic likelihood function as follows
\begin{align}
\argmax_{\mu^*, \sigma \geq 0, \mu_V, \sigma_V \geq 0} 
& \ 
\ln\Big\{ \mathscr{L}(\mu^*, \sigma, \mu_V, \sigma_V) \Big\} \nonumber \\
= & 
\ln\Bigg\{\prod_{j = 1}^{\widetilde{n}} \sum_{k=0}^{\infty} \mathbb{P}(n_t = k) f(\widetilde{z}_j | n_t) \Bigg\},
\end{align}
where $\widetilde{n}$ is the number of observations, the density $f(\widetilde{z}_j)$ is the sum of the conditional probabilities density $f(\widetilde{z}_j | n_t)$ weighted by the probability of the number of jumps $\mathbb{P}(n_t)$. This is an infinite mixture of normal variables, and there is usually one price jump if $\Delta t$ is small. Therefore, the estimation becomes:
\begin{equation}
\label{eq:para_estimation_mle}
\argmax_{\sigma \geq 0, \mu_V, \sigma_V \geq 0} \ln\Bigg\{\prod_{j = 1}^{\widetilde{n}} 
\Big((1 - \lambda \Delta t) f_1(\widetilde{z}_j) + \lambda \Delta t f_2(\widetilde{z}_j)
\Big) \Bigg\},
\end{equation}
where $f_1(\widetilde{z}_j)$ is the density of $\mathbf{N}\big((\mu - \frac{1}{2} \sigma^2 ) \Delta t, \sigma^2 \Delta t \big)$, and $f_2(\widetilde{z}_j)$ is the density of $\mathbf{N}\big((\mu - \frac{1}{2} \sigma^2 ) \Delta t + \mu_V, \sigma^2 \Delta t + \sigma^2_V \big)$. 


\subsection{Option Pricing Results and Revenue Analysis}

Fig.~\ref{fig:jump_price_example} presents examples of pricing an average price advertising option written on the keyword \lq\lq{}panasonic dmc\rq\rq{} from Google UK dataset. Each time step represents a day and the current time is indexed by 0. Therefore, the time indexes of the training set are $-59, \cdots, 0$, and the time indexes of the test set are $1, \cdots, 60$. Daily spot market prices of the first advertising position on the search engine result page exhibit frequent jumps in both training and test sets. The histogram of log change rates in the training set tends to have a higher peak than the normal distribution. By estimating the parameters of different jump-diffusion stochastic models based on the training data, we then simulate the paths of underlying spot market prices for the future period. Each path is generated by lines 6-12 in Algorithm~\ref{algo:option_pricing_mc}. The generated paths are used to calculate the option payoffs, and the average price is calculated based on the prices between time steps 30 and 60. We use the geometric mean (i.e., $\gamma=0$) to compute the average price so that we can compare the option prices calculated from Monte Carlo simulation and our derived pricing formula in Eq.~(\ref{eq:option_pricing_special_case}). Other model parameters are set as follows: $c = 0.2$, $\widetilde{c}=0.2$, $r=0.1$, $K = 0.75 X_0$. Fig.~\ref{fig:jump_price_example} shows that the generated price paths from three jump-diffusion models are different, which further affect the option price. For the log-normal jumps, we see the option price calculated in Eq.~(\ref{eq:option_pricing_special_case}) lies in the 95\% confidence interval of the price computed from Monte Carlo simulation. 

\begin{table}[tp]
\centering
\caption{Market performance and pricing specification.}
\label{tab:market_performance_and_pricing_specifications}
\begin{tabular}{l|l}
\hline
Description & Setting\\
\hline
Bull market & $X_0 \leq \frac{1}{\widetilde{m} + m} \sum_{i = 1}^{\widetilde{m} + m} X_i$\\
Bear market & $X_0 > \frac{1}{\widetilde{m} + m} \sum_{i = 1}^{\widetilde{m} + m} X_i$
\\
ITM & $K=0.75X_0$ \\
ATM & $K=X_0$ \\
OTM & $K=1.25X_0$ \\
\hline
\end{tabular}
\vspace*{-12pt}
\end{table}

\begin{table*}[t]
\centering
\caption{Comparison of revenues from selling advertising options and from advertising auctions.}
\label{tab:revenue_analysis}
\begin{tabular}{l|lll|lll}
\hline
\multicolumn{1}{c|}{Dataset} & \multicolumn{3}{c|}{Bull market 
} & \multicolumn{3}{c}{Bear market
}\\
\cline{2-7}
& ITM & ATM & OTM & ITM & ATM & OTM \\
\hline
\multicolumn{7}{l}{Log-normal jumps (explicit solution)} \\
\hline
SSP & 0.00\% (-83.89\%) &  0.00\% (-79.83\%) &  8.33\% (-75.73\%) &  100.00\% (214.46\%)   & 100.00\% (261.64\%)  &  100.00\% (309.17\%) \\
Google UK & 22.22\% (-30.63\%) & 22.22\% (-17.54\%) & 33.33\% (-4.41\%) &  96.97\% (272.85\%)  & 100.00\% (370.00\%)  & 100.00\% (467.29\%) \\
Google US & 3.84\% (-36.77\%) &  26.92\% (-22.06\%) &  53.85\% (-7.12\%) &  73.07\% (98.23\%)  & 100.00\% (146.81\%)  & 100.00\% (195.76\%) \\
\hline
\multicolumn{7}{l}{Log-normal jumps (Monte Carlo simulation)} \\
\hline
SSP & 0.00\% (-77.74\%) &  8.33\% (-74.07\%) &  8.33\% (-70.29\%) &  100.00\% (75.66\%)  &  100.00\% (122.06\%)  &  100.00\% (168.92\%) \\
Google UK  & 0.00\% (-48.46\%) & 22.22\% (-36.25\%) & 33.33\% (-23.78\%) &  96.97\% (282.04\%)  &  100.00\% (374.75\%)  &  100.00\% (468.55\%) \\
Google US & 21.79\% (-26.51\%) &  39.74\% (-14.15\%) &  57.69\% (-0.86\%) &  76.92\% (107.63\%)  & 100.00\% (152.23\%)  &  100.00\% (198.43\%) \\
\hline
\multicolumn{7}{l}{Log-ADE jumps (Monte Carlo simulation)} \\
\hline
SSP & 0.00\% (-84.50\%) &  0.00\% (-80.64\%) & 0.00\% (-76.63\%) &  100.00\% (101.71\%)  &  100.00\% (143.82\%)  &  100.00\% (187.67\%) \\
Google UK & 33.33\% (-26.01\%) & 33.33\% (-17.46\%) & 33.33\% (-9.65\%) &  83.33\% (255.44\%)  &  84.84\% (353.42\%)  & 84.84\% (453.39\%) \\
Google US & 16.67\% (-30.34\%) & 24.36\% (-20.26\%) &  44.87\% (-8.24\%) &  61.53\% (71.50\%)  & 100.00\% (114.23\%)  &  100.00\% (163.57\%) \\
\hline
\multicolumn{7}{l}{Log-laplacian jumps (Monte Carlo simulation)} \\
\hline
SSP & 0.00\% (-82.22\%) &  8.33\% (-78.45\%) &  8.33\% (-74.56\%) &  100.00\% (74.80\%)  & 100.00\% (121.34\%)  & 100.00\% (168.39\%) \\
Google UK  & 0.00\% (-54.12\%) & 11.11\% (-41.62\%) &  33.33\% (-28.93\%) &  81.82\% (240.84\%)  &  81.82\% (334.64\%)  &  81.82\% (429.33\%) \\
Google US & 2.56\% (-40.95\%) & 17.95\% (-27.79\%) &  41.02\% (-13.82\%) & 65.38\% (69.07\%)  &  100.00\% (115.04\%)  &  100.00\% (162.38\%) \\
\hline
\multicolumn{7}{p{6.75in}}{
The number outside the round brackets represents the percentage of advertising slots which have revenue increase and the number in the round brackets represents the average revenue change of slots under that group.
}
\end{tabular}
\vspace*{-12pt}
\end{table*}

We now examine the effects of the proposed advertising options on the seller's revenue. Recall that an option buyer will exercise the option in the future if he thinks the exercise price is less than what he pays in real-time auctions, otherwise, he will join auctions. Therefore, different combinations of market performance and pricing specification should taken into account. For the former, we simply consider bull and bear markets. A bull market describes the situation that the average spot market price in the future is equal to or higher than its present value while a bear market means the market is going down. Each advertising option can also be priced under three different specifications~\cite{Wilmott_2006_1}: in the money (ITM), at the money (ATM) and out of the money (OTM). Here an ITM option means the exercise price is less than the current spot market price when we price it. ATM and OTM options then represent the situations that the exercise price is equal or higher than the current spot market price, respectively. Pricing specification affects the computed option price, which can further affect the seller's revenue. Table~\ref{tab:market_performance_and_pricing_specifications} gives our settings of market performance and pricing specification. In our experiments, 66.67\% of slots in the SSP dataset, 12.00\% of slots in Google UK dataset, and 75.00\% of slots in Google US dataset, are classified into the bull market. 

Table~\ref{tab:revenue_analysis} shows the overall results of the revenue change for all advertising slots in our datasets. An inventory from a slot can be sold either through an advertising option or in real-time auctions. For the inventory, the revenue change of a slot is defined as a ratio $\frac{\textrm{Revenue}^{\textrm{Option}}-\textrm{Revenue}^{\textrm{Auction}}}{\textrm{Revenue}^{\textrm{Auction}}}$. $\textrm{Revenue}^{\textrm{Auction}}$ is the average spot market price during the period $[t_{\widetilde{m}}, t_{\widetilde{m}+m}]$. $\textrm{Revenue}^{\textrm{Option}}$ is sum of the option price and: (i) the exercise price $K$ if the option is exercised at time step $t_{\widetilde{m}}$; or (ii) the average spot market price during the period $[t_{\widetilde{m}}, t_{\widetilde{m}+m}]$ if the option is not exercised because the inventory can be auctioned off in real time. Whether an option would be exercised depends on its buyer's estimation of his payoff when he is at time step $t_{\widetilde{m}}$. We assume he can well estimate the future market so we use the test data in the period $[t_{\widetilde{m}}, t_{\widetilde{m}+m}]$ to calculate the option payoff. If the payoff is larger than zero, the option will be exercised otherwise not exercised. Table~\ref{tab:revenue_analysis} shows, if there is a bull market in the future, selling advertising options is not a good strategy for the seller because auctions are more profitable. Spot market prices are high in the bull market so option buyers will exercise their purchased options to hedge price risk. If the seller needs to sell advertising options for the bull market, OTM and ATM options are better choices than ITM options. This is because the exercise price is higher. If there is a bear market in the future, selling advertising options can significantly increase the seller's revenue. This is because spot market prices are cheap and option buyers will not exercise the purchased options. They will join advertising auctions instead so the seller's increased revenue is mainly contributed by the option prices. As we use the average price for calculating the option payoff, the effect of the exercise price on the option price is not very sensitive. Therefore, OTM and ATM options might still be better choices than ITM options in the bear market because of the higher exercise price.

\section{Concluding Remarks}
\label{sec:conclusion}

In this paper, we propose a new advertising option pricing framework. The option payoff is based on the power mean of the underlying spot market prices from a specific advertising slot over a future period. Therefore, the option is path-dependent and it addresses the biased option payoff calculation problem. We use a jump-diffusion stochastic process to model the underlying spot market prices over time, which allows discontinuities in price evolution. We obtain a general option pricing solution using Monte Carlo simulation and derive an explicit pricing formula for a special case. The latter is also a generalization of several option pricing models in the previous related studies~\cite{Black_1973, Merton_1973, Merton_1976,Zhang_1998,Kou_2002,Chen_2015_1}. In addition, our datasets cover both display advertising and sponsored search, and a set of stylized facts which are common to a wide set of online advertisements is summarized. To the best of our knowledge, it is the very first comprehensive summary of empirical properties of the spot market prices in online advertising auctions. 

Our study has three limitations. First, the volatility term in the underlying jump-diffusion stochastic process is a constant. Although our empirical findings show that volatility clustering is not a property for many advertisements, it would be good if we could further discuss a case or situation that the volatility term is uncertain such as the stochastic-volatility jump-diffusion model. This can be a future direction. Second, although the proposed underlying jump-diffusion stochastic process can possibly reproduce cyclical patterns based on the homogeneous Poisson process, we can further investigate a more accurate model using time series analysis and signal processing techniques. The challenge is how to use it for pricing advertising options which also rules out arbitrage opportunities. Third, capacity issue is not considered. In this paper, we assume that a media seller has a good estimation of future inventories and rationally sells them in advance via advertising options. Discussing capacity will include a game-theoretical analysis of the combined strategies of both buy-side and sell-side markets. Given the estimated capacity, penalty also can be added into the option pricing. 


\bibliographystyle{ieeetr}
\bibliography{mybib_long}



\section*{Acknowledgment}

This work is supported by the National Research Foundation, Prime Minister's Office, Singapore under its International Research Centre in Singapore Funding Initiative.


%
%

\bigskip

{
\noindent\textbf{Bowei Chen} (M'16) is an Assistant Professor at the Adam Smith Business School of the University of Glasgow, UK. He received a PhD in Computer Science from the University College London, and works in the cross-sections among machine learning, data science and business studies. His research interest lies in the applications of probabilistic modeling and deep learning in marketing, finance and information systems.  

\noindent\textbf{Mohan Kankanhalli} (F'14) received the BTech degree from IIT Kharagpur, Kharagpur, India, and the MS and PhD degrees from the Rensselaer
Polytechnic Institute, Troy, NY. He is a professor with the School of Computing, National University of Singapore, Singapore. He is the director with the SeSaMe Centre and also the Dean, School of Computing at NUS. His research interests include multimedia systems and multimedia security. He is active in the Multimedia Research Community. He was the ACM SIGMM director of Conferences from 2009 to 2013. He is on the editorial boards of several journals. He is a fellow of the IEEE.
\vfill
}

\end{document}